\documentclass[journal]{IEEEtran}
\ifCLASSINFOpdf

\else

\fi

\usepackage{mathrsfs}
\usepackage{amsmath}
\usepackage{amsfonts}
\usepackage{amsthm}
\usepackage{amssymb}
\usepackage{graphicx}
\usepackage{subfigure}
\usepackage{indentfirst}
\usepackage{array}
\usepackage{cite}
\usepackage{enumerate}
\usepackage{bm}
\usepackage[linesnumbered,ruled,vlined]{algorithm2e}
\usepackage{algorithmic}
\usepackage{multirow}
\usepackage{epstopdf}
\usepackage{verbatim}
\usepackage{stfloats}
\usepackage{color}
\usepackage{caption}
\usepackage{booktabs}
\usepackage{footnote}

\newtheorem{theorem}{\bf{Theorem}}
\newtheorem{lemma}{\bf{Lemma}}

\newtheorem{definition}{\bf{Definition}}
\newtheorem{remark}{\bf{Remark}}
\SetKwComment{Comment}{//}{}

\begin{document}

\title{\huge{Construction Methods Based on Minimum Weight Distribution for Polar Codes with Successive Cancellation List Decoding}}
\author{Jinnan Piao, Dong Li, Jindi Liu, Xueting Yu, Zhibo Li, Ming Yang, and Peng Zeng
\thanks{
This work is supported in part by the National Key R\&D Program of China under Grant 2022YFB3207400, in part by the National Natural Science Foundation of China under Grant 62201562 and 92267301, in part by the Liaoning Provincial Natural Science Foundation of China under Grant 2024--BSBA--51, and in part by the Fundamental Research Project of SIA under Grant 2022JC1K08. (\emph{Corresponding Author: Dong Li})
}
\thanks{
The authors are with the State Key Laboratory of Robotics, Shenyang Institute of Automation, Chinese Academy of Sciences, Shenyang 110016, China, and with the Key Laboratory of Networked Control Systems,  Chinese Academy of Sciences, Shenyang 110016, China.
(e-mail: piaojinnan@sia.cn; lidong@sia.cn; liujindi@sia.cn; yuxueting@sia.cn; lizhibo@sia.cn; yangming@sia.cn; zp@sia.cn).
}
}

\maketitle
\begin{abstract}

Minimum weight distribution (MWD) is an important metric to calculate the first term of union bound called minimum weight union bound (MWUB). In this paper, we first prove the maximum likelihood (ML) performance approaches MWUB as signal-to-noise ratio (SNR) goes to infinity and provide the deviation when MWD and SNR are given. Then, we propose a nested reliability sequence, namely MWD sequence, to construct polar codes independently of channel information. In the sequence, synthetic channels are sorted by partial MWD which is used to evaluate the influence of information bit on MWD and we prove the MWD sequence is the optimum sequence evaluated by MWUB for polar codes obeying partial order. Finally, we introduce an entropy constraint to establish a relationship between list size and MWUB and propose a heuristic construction method named entropy constraint bit-swapping (ECBS) algorithm, where we initialize information set by the MWD sequence and gradually swap information bit and frozen bit to satisfy the entropy constraint. The simulation results show the MWD sequence is more suitable for constructing polar codes with short code length than the polar sequence in 5G and the ECBS algorithm can improve MWD to show better performance as list size increases.

\end{abstract}


\IEEEpeerreviewmaketitle

\section{Introduction}

\subsection{Relative Research}

\IEEEPARstart{P}{olar} codes, invented by Ar{\i}kan, have been proved to achieve the capacity of arbitrary binary-input discrete memoryless channels (B-DMCs) as the code length goes to infinity with the successive cancellation (SC) decoding \cite{arikan}.
However, the performance of polar codes is weak at short to moderate code lengths under SC decoding. To improve the block error rate (BLER) performance, the successive cancellation list (SCL) decoding \cite{niuscl,talvardyscl} is proposed.
Since polar codes demonstrate advantages in error performance and other attractive application prospects, they have been adopted as the coding scheme for the control channel of the enhanced Mobile Broadband (eMBB) service category in the fifth generation wireless communication systems (5G) \cite{3GPP_5G_polar}.

According to the polarization effect, the original channels of polar codes are transformed into synthetic channels
with different reliability. The construction of polar codes is selecting the most reliable synthetic channels to transmit information bits and others to transmit frozen bits.
The widely used design principle is to minimize the BLER under the SC decoding. The Bhattacharyya parameter \cite{arikan} is the first construction method to precisely calculate the mutual information of each synthetic channel in the binary erasure channel (BEC).
For other B-DMCs, the density evolution (DE) algorithm, initially proposed in \cite{DE} and improved in \cite{TalVardy}, tracks the probability distribution of the logarithmic likelihood ratio (LLR) of each synthetic channel and provides theoretical guarantee on the estimation accuracy with a high computational cost.
For the binary-input additive white Gaussian noise (BI-AWGN) channel, the Gaussian approximation (GA) algorithm \cite{GA,GA_DAI} approximates the probability distribution of the LLR as Gaussian distribution and gives reliability evaluation with limited complexity.
For fading channels, the polar spectrum \cite{PolarSpectrum, PolarSpectrumFastFading} is proposed to derive the upper bound of the error probability of the synthetic channels to construct polar codes.
Generally, the above algorithms are known as channel-dependent construction, since the reliabilities of synthetic channels are derived from the original channel parameters.

From the viewpoint of system design, the construction should be independent of channel conditions to facilitate the practical application of polar codes. Sch{\"u}rch \cite{PartialOrder} first finds the partial reliability order is invariant in any B-DMC and proposes the concept of partial order (PO).
Then, He \emph{et al.} \cite{PW} propose the polarization weight (PW) algorithm exploiting the index expansion of synthetic channels to design a universal reliability sequence without the channel parameters. However, the PW sequence has performance loss for polar codes with the code length $N \ge 1024$.
In 5G, polar codes are constructed by a fixed universal polar sequence \cite{5GDesignPolar} for all the code configurations by computer searching with the maximum polar sequence length $1024$.

To improve the performance of polar codes under SCL decoding,
several construction algorithms based on machine learning \cite{ConstructionAI,ConstructionGenetic,ConstructionRL} are first proposed and have remarkable performance with extreme complexity.
Then, Mondelli \emph{et al.} \cite{RMpolar} show that the performance of polar codes with SCL decoding depends on its maximum likelihood (ML) performance.
To minimize the ML performance of short precoded polar codes,
Miloslavskaya \emph{et al.} \cite{DesShortPolarSCL, RecDesPolarSCL} propose a construction method under SCL decoding to optimize the precoding matrix with a large number of code search.

A method to evaluate the ML performance is calculating the union bound by the distance spectrum \cite{LinShuBook}.
Nevertheless, enumerating the distance spectrum has exponential complexity and it is almost impossible for long code length.
In the high signal-to-noise ratio (SNR) region, the first term of union bound calculated by the minimum weight distribution (MWD) is dominant in the union bound \cite{CRCdesign}, namely minimum weight union bound (MWUB) in this paper.
To analyze the MWD of polar codes, the SCL based methods \cite{ADSCL,dsliu,CRCdesign} with excessively large list size are used to enumerate the codewords. Then, a polynomial-time method is proposed in \cite{calculate_MWD} to calculate the MWD of polar codes obeying the PO.
For any construction of polar codes, the sphere constraint based enumeration methods \cite{sphereMWD} are proposed to analyze the MWD.
The formation of polar/polarization-adjusted convolutional (PAC) \cite{PAC} codes is proposed in \cite{NondreasingCodeMWD} to calculate MWD.
The average distance spectrum of pre-transformed polar codes is researched in \cite{AWS1, AWS2}.
The MWD is also used to optimize the polar-like codes to improve the BLER performance of SCL decoding \cite{NondreasingCodeMWD, PolarOpt0025, yuan2019polar}. Thus, MWD is a critical metric to evaluate the ML decoding performance and optimize polar codes under SCL decoding.

\subsection{Motivation}

The SCL decoding is widely used for polar codes and its performance can approach the ML performance with limited list size \cite{talvardyscl}. Meanwhile, the MWD is an effective metric to evaluate ML performance. Thus, in this paper, we focus on the construction methods based on MWD to improve the performance of polar codes under SCL decoding.

\subsection{Main Contributions}

The main contributions of this paper are summarized as follows.

\begin{enumerate}
  \item We show the ML performance approaches the MWUB as the SNR goes to infinity by proving that both the upper bound and the lower bound of the ML performance approach the MWUB.
      Then, we derive the normalized deviation of the MWUB from the ML performance when the MWD and the SNR are given by introducing a looser union bound which concentrates the full distance spectrum on the small weight distribution.

  \item We introduce a new concept on the synthetic channels of polar codes, named partial MWD, which is used to evaluate the influence of each synthetic channel on the MWD when the information bit is transmitted in the synthetic channel.
      Then, based on the partial MWD, we order the synthetic channels and obtain a nested construction sequence independently of channel information, called MWD sequence.
      Finally, we prove that the MWD sequence is the optimum sequence evaluated by MWUB for polar codes obeying PO.

  \item A heuristic and greedy entropy constraint bit-swapping (ECBS) algorithm is proposed to improve the performance of polar codes under SCL decoding with limited list size.
      To design the ECBS algorithm, we establish a relationship between the list size and the MWUB by the entropies of the synthetic channels transmitting information bits, namely entropy constraint.
      In the ECBS algorithm, the information set is initialized by the MWD sequence. Then, we set the swapping range by the partial MWD and swap information bit and frozen bit greedily to make the information set satisfy the entropy constraint.

\end{enumerate}

The simulation results show that the proposed MWD sequence is suitable for constructing polar codes for short code length and has about $0.8$dB performance gain for code length $256$ and list size $16$ at code rate $0.5$ and BLER $10^{-3}$ compared with the polar sequence in 5G.
Then, the $\left(128, 64\right)$ polar code constructed by MWD sequence with $6$-bit CRC and list size $128$ has $0.22$dB performance gap at BLER $10^{-4}$ compared with the finite blocklength capacity.
Finally, the simulation results show that the performance of polar codes constructed by the proposed ECBS algorithm approaches the MWUB when the entropy constraint is satisfied.

The remainder of the paper is organized as follows. Section
II describes the preliminaries of polar codes, MWD and SCL decoding. The properties of MWUB are shown in Section III.
In Section IV, we introduce the concept of partial MWD and propose the MWD sequence. The ECBS algorithm is proposed in Section V. Section VI shows the performance of polar codes constructed by the MWD sequence and the ECBS algorithm. Section VII concludes this paper.

\section{Notations and Preliminaries}

\subsection{Notation Conventions}

In this paper, the lowercase letters, e.g., $x$, are used to denote scalars. The bold lowercase letters, e.g., ${\bf{x}}$, are used to denote vectors.
Notation ${{\bf x}_i^j}$ denotes the subvector $(x_i,\cdots,x_j)$ and $x_i$ denotes the $i$-th element of ${\bf{x}}$.
The sets are denoted by calligraphic characters, e.g., $\cal{X}$, and the notation $|\cal{X}|$ denotes the cardinality of $\cal{X}$.
The bold capital letters, e.g., $\bf{X}$, are used to denote matrices.
Furthermore, we write ${\bf{F}}^{\otimes n}$ to denote the $n$-th Kronecker power of $\bf{F}$.
Throughout this paper, $\bf 0$ and $\bf 1$ mean an all-zero vector and an all-one vector, respectively, and $\log \left(  \cdot  \right)$ means ``base 2 logarithm''.

\subsection{Polar Codes}

Polar codes depend on the polarization effect \cite{arikan} of the matrix
${{\bf F} = \left[
\begin{smallmatrix}
1&0\\
1&1
\end{smallmatrix}
\right]}$.
For an $(N,K)$ polar code with code length $N = 2^n$ and code rate $R = K/N$, the polarization effect generates $N$ synthetic channels $W_N^{\left(i\right)}, i = 0,1,\cdots,N-1$.
Each synthetic channel has different reliability $U\left(W_N^{\left(i\right)}\right)$ and the information bits are transmitted in the $K$ most reliable synthetic channels.
Hence, the information set of polar codes defined by ${\cal A}$ with cardinality $|{\cal A}|=K$ is composed of the indices of the $K$ most reliable synthetic channels. Then, the frozen set ${\cal A}^c$ with cardinality $|{\cal A}^c|=N-K$ is a complementary set of ${\cal A}$.
The codeword ${\bf c}$ of the polar code is calculated by ${\bf c}  = {\bf u}{\bf G}$, where ${\bf u}$ is an $N$-length information sequence and ${\bf G}$ is ${\bf F}^{\otimes n}$.
The information sequence ${\bf u}$ is generated by assigning $u_i$ to information bit if $i \in {\cal A}$, and assigning $u_i$ to $0$ if $i \in {\cal A}^c$.

Without loss of generality, the AWGN channel and the binary phase shift keying (BPSK) modulation are considered in this paper. Thus, each coded bit $c_i \in \left\{0,1\right\}$ is modulated into the transmitted signal by $s_i = 1 - 2c_i$. Then, the received sequence is ${\bf y} = {\bf s} + {\bf n}$, where $n_i$ is i.i.d. AWGN with zero mean and variance $\sigma^2$.

\subsection{Minimum Weight Distribution}

The distance spectrum of an $(N, K)$ binary linear block code, denoted by $A_d$, is the number of codewords of the code with the Hamming weight $d$. The pairwise error probability between two codewords modulated by BPSK differing in $d$ positions and coherently detected in the AWGN channel is $Q\left(\sqrt{\frac{{2dRE_b}}{{{N_0}}}}\right)$, where $E_b$ is the energy of the transmitted bit, $N_0$ is the one-sided power spectral density of AWGN and
\begin{equation}\label{Q_function}
Q(x)=\frac{1}{{\sqrt {2\pi }  }}\int_x^\infty  {{e^{ - \frac{{{t^2}}}{2}}}dt}
\end{equation}
is the probability that a Gaussian random variable with zero mean and unit variance exceeds the value $x$. Assuming that an all-zero codeword $\bf 0$ is transmitted, the union bound \cite{LinShuBook} of ML performance is
\begin{equation}\label{union_bound}
P_e
\le \sum\limits_{d = {d_{\min }}}^N {{A_d}Q\left( {\sqrt {\frac{{2dR{E_b}}}{{{N_0}}}} } \right)} .
\end{equation}

Then, since the MWD (i.e., $d_{\min}$ and $A_{d_{\min}}$) is the main factor influencing the ML performance in the high SNR region,
(\ref{union_bound}) is approximated as the first term of union bound \cite{LinShuBook, CRCdesign,sphereMWD,NondreasingCodeMWD}, i.e.,
\begin{equation}\label{union_bound_a}
P_e \approx {{A_{d_{\min}}}Q\left( {\sqrt {\frac{{2{d_{\min}}R{E_b}}}{{{N_0}}}} } \right)},
\end{equation}
where $d_{\min}$ is the minimum Hamming weight of the linear block code and ${A_{d_{\min}}}$ is the number of the codewords with $d_{\min}$.
The right-hand side (RHS) of \eqref{union_bound_a} is called MWUB in this paper.
Rigorously, \eqref{union_bound_a} cannot indicate that the ML performance approaches the MWUB as the SNR goes to infinity.

\subsection{Successive Cancellation List Decoding}

Polar codes can be decoded by the SC decoding algorithm with the decoding complexity $O(N\log N)$ \cite{arikan}.
Although SC decoding can achieve the channel capacity with the infinite code length, its performance is unsatisfactory for short or medium code lengths.

Thus, the SCL decoding is proposed in \cite{talvardyscl}, which recursively computes $P\left({\widehat {\bf u}}_0^i | {\bf y} \right)$ for $i=0,1,\cdots,N-1$ with the process similar to the SC decoding except keeping at most $L$ survival paths.
Let ${\mathcal U}_{i} \subseteq \left\{0,1\right\}^{i+1}, i = 0,\cdots, N-1$ be the subset including the $L$ survival paths at the $i$-th decoding step with the path metric

\begin{equation}
P\left({\widehat {\bf u}}_0^i | {\bf y} \right) = P\left({\widehat {\bf u}}_0^{i-1} | {\bf y} \right)P\left({{\widehat u}}_i | {\bf y}, {\widehat {\bf u}}_0^{i-1} \right).
\end{equation}
When $u_{i+1}$ is an information bit, the $L$ survival paths in ${\mathcal U}_{i}$ are split into $2L$ paths with attempting calculating
$P\left({{\widehat u}}_{i+1} = 1 | {\bf y}, {\widehat {\bf u}}_0^{i} \right)$ and
$P\left({{\widehat u}}_{i+1} = 0 | {\bf y}, {\widehat {\bf u}}_0^{i} \right)$.
Then, ${\mathcal U}_{i+1}$ is decided by selecting the $L$ most likely paths with larger $P\left({\widehat {\bf u}}_0^{i+1} | {\bf y} \right)$ from the $2L$ split paths. When $u_{i+1}$ is frozen bit, the $L$ survival paths in ${\mathcal U}_{i}$ are simply extended with the correct frozen bit.

\section{Properties of MWUB}\label{SectionMWUB}

In this section, we derive the ML performance approaches the MWUB as the SNR goes to infinity and provide the normalized deviation of the MWUB from the ML performance.

To derive the limit of ML performance, we first introduce the lower bound \cite[eq. (29)]{KouniaBound} as
\begin{equation}\label{KB}
\begin{aligned}
P_e \ge& b^*Q\left( {\sqrt {\frac{{2{d_{\min}}R{E_b}}}{{{N_0}}}} } \right) - \\ &\frac{b^*\left(b^*-1\right)}{2}
\Psi\left(\frac{1}{2},\sqrt {\frac{{2{d_{\min}}R{E_b}}}{{{N_0}}}},\sqrt {\frac{{2{d_{\min}}R{E_b}}}{{{N_0}}}}\right),
\end{aligned}
\end{equation}
where
\begin{equation}
b^* = \min\left\{\frac{1}{2}+\frac{Q\left( {\sqrt {\frac{{2{d_{\min}}R{E_b}}}{{{N_0}}}} } \right)}{\Psi\left(\frac{1}{2},\sqrt {\frac{{2{d_{\min}}R{E_b}}}{{{N_0}}}},\sqrt {\frac{{2{d_{\min}}R{E_b}}}{{{N_0}}}}\right)}, A_{d_{\min}}\right\}
\end{equation}
and
\begin{equation}\label{Eq2DgaussianPro}
\begin{aligned}
\Psi\left(\rho,x',y'\right)
= \frac{1}{2\pi\sqrt{1-\rho^2}} \int\limits_{x'}^{\infty}{\int\limits_{y'}^{\infty}{e^{-\frac{x^2-2\rho xy+y^2}{2\left(1-\rho^2\right)}}}}dxdy.
\end{aligned}
\end{equation}

With the lower bound \eqref{KB} and the union bound \eqref{union_bound}, the limit of ML performance is provided in Lemma \ref{LemmaMLlimit}.
\begin{lemma}\label{LemmaMLlimit}
As $\frac{E_b}{N_0}$ goes to infinity, $P_e$ approaches the MWUB, i.e.,
\begin{equation}
\lim_{\frac{E_b}{N_0}\rightarrow\infty}P_e = A_{d_{\min}}Q\left( {\sqrt {\frac{{2{d_{\min}}R{E_b}}}{{{N_0}}}} } \right).
\end{equation}
\end{lemma}

\begin{proof}
The proof is equivalent to both the lower bound \eqref{KB} and the union bound \eqref{union_bound} approaching the MWUB as $\frac{E_b}{N_0}\rightarrow\infty$.

For the lower bound, given two random variables ${\mathsf X} \sim {\mathcal N}(0,1)$ and ${\mathsf Y} \sim {\mathcal N}(0,1)$ and the correlation coefficient $\rho$, the probability density function of a 2-dimensional Gaussian distribution is
\begin{equation}
p\left(x, y\right) = \frac{1}{2\pi\sqrt{1-\rho^2}} e^{-\frac{x^2-2\rho xy+y^2}{2\left(1-\rho^2\right)}}
\end{equation}
Then, we transform \eqref{Eq2DgaussianPro} into
\begin{equation}\label{Eq2DgaussianProTransform}
\begin{aligned}
&\Psi\left(\rho,x',y'\right)
= \int\limits_{x'}^{\infty}{\int\limits_{y'}^{\infty}{p\left(x, y\right)}}dxdy\\
&=\int\limits_{x'}^{\infty}\left(\int\limits_{-\infty}^{\infty}{p\left(x, y\right)}dy\right)dx - \int\limits_{x'}^{\infty}{\int\limits_{-\infty}^{y'}{p\left(x, y\right)}}dxdy \\
&=\int\limits_{x'}^{\infty}p\left(x\right)dx - \Xi\left(\rho,x',y'\right)\\
&=Q\left(x'\right) - \Xi\left(\rho,x',y'\right),
\end{aligned}
\end{equation}
where
\begin{equation}
\Xi\left(\rho,x',y'\right) = \int\limits_{x'}^{\infty}{\int\limits_{-\infty}^{y'}{p\left(x, y\right)}}dxdy.
\end{equation}

Given $X_{d_{\min}} = {\sqrt {\frac{{2{d_{\min}}R{E_b}}}{{{N_0}}}} }$, we have
\begin{equation}\label{EqMLLowerBoundLimit}
\begin{aligned}
&\lim_{X_{d_{\min}}\rightarrow\infty}\left|\frac{b^*Q\left( X_{d_{\min}} \right) - \frac{b^*\left(b^*-1\right)}{2}
\Psi\left(\frac{1}{2},X_{d_{\min}},X_{d_{\min}}\right)}{{A_{d_{\min}}}Q\left( X_{d_m} \right)}\right| \\
&=\left|1 - \frac{A_{d_{\min}} - 1}{2Q\left( X_{d_{\min}} \right)}\left(Q\left( X_{d_{\min}} \right) - \Xi\left(\frac{1}{2},X_{d_{\min}},X_{d_{\min}}\right)  \right)\right| \\
&=\left|1 - \frac{A_{d_{\min}} - 1}{2Q\left( X_{d_{\min}} \right)}\left(Q\left( X_{d_{\min}} \right) - Q\left( X_{d_{\min}} \right)  \right)\right| = 1.
\end{aligned}
\end{equation}
Hence, the RHS of \eqref{KB} approaches the MWUB as $\frac{E_b}{N_0}\rightarrow\infty$.

For the union bound, we introduce the limit of $Q\left(x\right)$ as
\begin{equation}
\lim_{x\rightarrow\infty}Q\left(x\right) = \frac{e^{-\frac{x^2}{2}}}{\sqrt{2\pi}x}.
\end{equation}
Then, we have
\begin{equation}\label{EqMLUnionBoundLimit}
\begin{aligned}
&\lim_{\frac{E_b}{N_0}\rightarrow\infty}
\left|\frac{\sum\limits_{d = {d_{\min }}}^N {{A_d}Q\left( {\sqrt {\frac{{2dR{E_b}}}{{{N_0}}}} } \right)}}{{A_{d_{\min}}}Q\left( {\sqrt {\frac{{2{d_{\min}}R{E_b}}}{{{N_0}}}} } \right)} \right| \\
&= \left|1+ \sum\limits_{d={{d}_{\min }+1}}^{N}{\frac{{{A}_{d}}\sqrt{{{d}_{\min }}}}{{{A}_{{{d}_{\min }}}}\sqrt{d}}{{e}^{\frac{\left( {{d}_{\min }}-d \right)R{{E}_{b}}}{{{N}_{0}}}}}} \right| = 1.
\end{aligned}
\end{equation}
Hence, the union bound approaches the MWUB as $\frac{E_b}{N_0}\rightarrow\infty$. With \eqref{EqMLLowerBoundLimit} and \eqref{EqMLUnionBoundLimit}, Lemma \ref{LemmaMLlimit} is proven.
\end{proof}

Then, to derive the normalized deviation of the MWUB from the ML performance, we provide a looser union bound by concentrating the full distance spectrum on the small weight distribution, i.e.,
\begin{equation}\label{EqLooserUnionBound}
\begin{aligned}
\sum\limits_{d = {d_{\min }}}^N {{A_d}Q\left( {\sqrt {\frac{{2dR{E_b}}}{{{N_0}}}} } \right)}
\le {A_{d_{\min}}}Q\left( {\sqrt {\frac{{2{d_{\min}}R{E_b}}}{{{N_0}}}} } \right) + &\\ \sum\limits_{d = {d_{\min }}+2, d_{\min }+4,\cdots,N}{{L_d}Q\left( {\sqrt {\frac{{2dR{E_b}}}{{{N_0}}}} } \right)},&
\end{aligned}
\end{equation}
where the Hamming weight of polar codes is even \cite{LinShuBook}.
When $d$ is even, we have
\begin{equation}
L_d = \min\left\{ \max\left\{0,2^K - A_{d_{\min}} - \sum\limits_{d' = {d_{\min }}+1}^{d-1}{L_{d'}}\right\}, \left( \begin{aligned}
  & N \\
 & d \\
\end{aligned} \right) \right\}
\end{equation}
and when $d$ is odd, we have $L_d = 0$.
Obviously, the RHS of \eqref{EqLooserUnionBound} is close to the MWUB as $\frac{E_b}{N_0}\rightarrow\infty$.

With \eqref{KB} and \eqref{EqLooserUnionBound}, the normalized deviation is provided in Lemma \ref{LemmaNormDeviation}.

\begin{lemma}\label{LemmaNormDeviation}
Given $P_{\tt UB}$, $P_{\tt MWUB}$, $P_{\tt LB}$ and $P_{\tt LUB}$ be the RHSs of \eqref{union_bound}, \eqref{union_bound_a}, \eqref{KB} and \eqref{EqLooserUnionBound}, respectively, the upper bound of the normalized deviation $\delta$ of $P_{\tt MWUB}$ from $P_e$ is
\begin{equation}
\delta = \left|\frac{P_e - P_{\tt MWUB}}{P_{\tt MWUB}}\right| \le \frac{P_{\tt UB} - P_{\tt LB}}{P_{\tt MWUB}} \le \frac{P_{\tt LUB} - P_{\tt LB}}{P_{\tt MWUB}}.
\end{equation}
\end{lemma}

\begin{proof}
Since $P_{\tt LB} \le P_e \le P_{\tt UB} \le P_{\tt LUB}$ and $P_{\tt LB} \le P_{\tt MWUB} \le P_{\tt UB} \le P_{\tt LUB}$, we have $\left|P_e - P_{\tt MWUB}\right| \le P_{\tt UB} - P_{\tt LB} \le P_{\tt LUB} - P_{\tt LB}$ and Lemma \ref{LemmaNormDeviation} is proven.\end{proof}

For an $\left(N, K\right)$ polar codes with $d_{\min}$ and $A_{d_{\min}}$, given the SNR $\frac{E_b}{N_0}$, we can calculate $\delta_1 = \frac{P_{\tt UB} - P_{\tt LB}}{P_{\tt MWUB}}$ and $\delta_2 = \frac{P_{\tt LUB} - P_{\tt LB}}{P_{\tt MWUB}}$
satisfying $\delta \le \delta_1 \le \delta_2$.
Thus, Lemma \ref{LemmaNormDeviation} describes the degree of the MWUB approaching the ML performance with fixed SNR.

Fig. \ref{FigMWUB} illustrates $P_{\tt UB}$, $P_{\tt MWUB}$, $P_{\tt LB}$, $P_{\tt LUB}$, $\delta_1$ and $\delta_2$ of $(64,32)$ polar code constructed by PW with $d_{\min} = 8$ and $A_{d_{\min}} = 664$.
The distance spectrum of the $(64,32)$ polar code is obtained by enumerating all the codewords.
In the above of Fig. \ref{FigMWUB}, we observe that $P_{\tt UB}$, $P_{\tt LB}$ and $P_{\tt LUB}$ approach $P_{\tt MWUB}$ as SNR increases.
The bottom of Fig. \ref{FigMWUB} shows $\frac{E_b}{N_0} = 4.17$dB at $\delta_1 = 1$ is less than $\frac{E_b}{N_0} = 11.95$dB at $\delta_2 = 1$. Hence, by $\delta_1$, the MWD is the dominant factor to determine the ML performance at BLER $10^{-3}$.
In comparison, $\delta_2$ shows larger required SNR, but calculating $\delta_2$ is easy since it is just related to MWD.

\begin{figure}[t]
\setlength{\abovecaptionskip}{0.cm}
\setlength{\belowcaptionskip}{-0.cm}
  \centering{\includegraphics[scale=0.69]{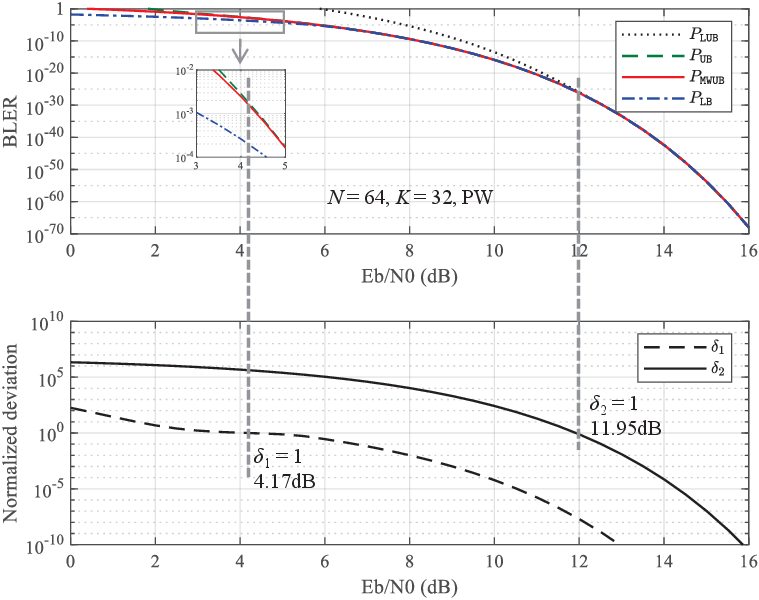}}
  \caption{$P_{\tt UB}$, $P_{\tt MWUB}$, $P_{\tt LB}$, $P_{\tt LUB}$, $\delta_1$ and $\delta_2$ of $(64,32)$ polar code constructed by PW with $d_{\min} = 8$ and $A_{d_{\min}} = 664$.}\label{FigMWUB}
  \vspace{-0em}
\end{figure}

\section{Minimum Weight Distribution Sequence}

In this section, we first describe the decreasing monomial codes and the calculation process of MWD.
Then, the concept of partial MWD is introduced to evaluate the influence of synthetic channel on MWD.
Finally, a nested MWD sequence is designed by the partial MWD.

\subsection{Decreasing Monomial Codes}

An $\left(N, K\right)$ polar code can be viewed as a decreasing monomial code with the monomial set ${\mathcal M}$ and the information subset ${\mathcal I}$ \cite{calculate_MWD}.
Each synthetic channel of polar codes corresponds to a monomial.
For the synthetic channel $W_N^{\left(i\right)}$,
the binary representation
of $i$ is
\begin{equation}\label{binary_i}
{\bf b}^i = \left[b_{n-1}^i, \cdots, b_1^i, b_0^i\right],
\end{equation}
where $b_{n-1}^i$ and $b_0^i$ are the most and the least significant bits, respectively.
Given the binary variable collection $\left\{x_0,x_1,\cdots,x_{n-1}\right\}$, the monomial $f_i$ corresponding to $W_N^{\left(i\right)}$ contains the variables with
the bit-positions of zero in ${\bf b}^i$, i.e.,
\begin{equation}\label{monomial_l}
f_i = \prod_{l=0}^{n-1}{x_l^{1 - b_l^i}} = \prod_{m = 1}^{\deg \left(f_i\right)}{x_{i_m}},
\end{equation}
where $i_m$ is the index of zero in ${\bf b}^i$ with $i_1 < \cdots < i_{\deg \left(f_i\right)}$
and $\deg \left(f_i\right)$ is the degree of $f_i$, i.e.,
\begin{equation}\label{degfi}
\deg \left(f_i\right) = n - \sum_{l=0}^{n-1}{b_l^i}.
\end{equation}

The monomial set ${\mathcal M}$ including all the $N$ monomials is
\begin{equation}\label{Mn}
{\mathcal M} = \left\{f_i : i=0,1,\cdots,N-1 \right\}
\end{equation}
and the information subset ${\mathcal I}$ is
\begin{equation}\label{monomial_I}
{\mathcal I} = \left\{f_i : i \in {\mathcal A} \right\}.
\end{equation}

In \cite{PartialOrder} and \cite{calculate_MWD}, the synthetic channels exhibit a PO ``$\preceq$'' with respect to their reliability, i.e., $f_i \preceq f_j$ means $W_N^{\left(i\right)}$ is more reliable than $W_N^{\left(j\right)}$. The PO for the monomials with the same degree is defined as follows.
\begin{definition}\label{def1_PO}
Given two monomials $f_i = \prod_{m=1}^{D}{{{x}_{i_m}}}$ and $f_j = \prod_{m=1}^{D}{{{x}_{j_m}}}$ with the same degree $D = \deg \left(f_i\right) = \deg \left(f_j\right)$, we have $f_i \preceq f_j$,
if and only if $i_m \le j_m, m = 1,2,\cdots,D$.
\end{definition}

The definition of the PO for the monomials with different degrees is as follows.
\begin{definition}\label{def2_PO}
Given two monomials $f_i = \prod_{m=1}^{\deg(f_i)}{{{x}_{i_m}}}$ and $f_j = \prod_{m=1}^{\deg(f_j)}{{{x}_{j_m}}}$ with $\deg(f_i) \neq \deg(f_j)$, we have $f_i \preceq f_j$, if and only if there is a divisor $f_j^*$ of $f_j$ making $f_i \preceq f_j^*$.
\end{definition}

A polar code is a decreasing monomial code with the monomials in ${\mathcal I}$ obeying the PO in Definition \ref{def1_PO} and Definition \ref{def2_PO}, which means that if $W_N^{\left(j\right)}$ is selected as the information channel, all the synthetic channels more reliable than $W_N^{\left(j\right)}$ satisfying the PO are also the information channels, i.e., if $f_i \preceq f_j$ and $f_j \in {\mathcal I}$, we have $f_i \in {\mathcal I}$.

For the decreasing monomial codes, a method to calculate the MWD is introduced in \cite{calculate_MWD}. Let ${\mathcal I}_r$ be the subset of ${\mathcal I}$ with the monomials having degree $r$, i.e.,
\begin{equation}
{\mathcal I}_r = \left\{f_i : \deg(f_i) = r, f_i \in {\mathcal I} \right\}
\end{equation}
and $r_{\max}^{\mathcal I}$ be the maximum degree of the monomials in ${\mathcal I}$, i.e.,
\begin{equation}
r_{\max}^{\mathcal I} = \max_{f_i \in {\mathcal I}}{\deg{(f_i)}}.
\end{equation}
The number of the minimum weight codewords with ${\mathcal I}$ is
\begin{equation}\label{CalMWDAdmin}
A_{d_{\min}}^{\mathcal I} = 2^{r_{\max}^{\mathcal I}} \times \sum_{f_i \in {\mathcal I}_{r_{\max}^{\mathcal I}}}{2^{\left|\lambda_{f_i}\right|}},
\end{equation}
where the minimum Hamming weight with the subset ${\mathcal I}$ is
\begin{equation}\label{CalMWDdmin}
d_{\min}^{\mathcal I} = 2^{n-r_{\max}^{\mathcal I}}
\end{equation}
and
\begin{equation}\label{LamedaCal}
\left|\lambda_{f_i}\right| =\frac{\deg(f_i)\left(1-\deg(f_i)\right)}{2} + \sum_{m=1}^{\deg(f_i)}{i_m}.
\end{equation}

\subsection{Partial MWD}

We propose a new concept, namely partial MWD, which is used to evaluate the influence of each synthetic channel on the MWD when the synthetic channel is selected to transmit information bit.

Assuming that existing an ordered construction sequence ${\bf q} = \left[q_0,q_1,\cdots,q_{N-1}\right]$ obeys the PO in Definition \ref{def1_PO} and Definition \ref{def2_PO} with $W_N^{(q_{k-1})}$ more reliable than $W_N^{(q_{k})}$, $k=1,2,\cdots,N-1$, given the subset
\begin{equation}
{\mathcal I}^{k} = \left\{f_{q_0},f_{q_1},\cdots,f_{q_{k}}\right\},
\end{equation}
according to \eqref{CalMWDAdmin} and \eqref{CalMWDdmin}, the difference of the MWD between ${\mathcal I}^{k}$ and ${\mathcal I}^{k-1}$ is as follows:
\begin{enumerate}
  \item If $\deg(f_{q_k}) = r_{\max}^{{\mathcal I}^{k-1}}$, we have
      $d_{\min}^{{\mathcal I}^{k}} = d_{\min}^{{\mathcal I}^{k-1}} = 2^{n-\deg(f_{q_k})}$
      and
     $ A_{d_{\min}}^{{\mathcal I}^k} - A_{d_{\min}}^{{\mathcal I}^{k-1}} = 2^{\deg(f_{q_k}) + \left|\lambda_{f_{q_k}}\right|}$.
  \item If $\deg(f_{q_k}) > r_{\max}^{{\mathcal I}^{k-1}}$, we have
      $d_{\min}^{{\mathcal I}^k} = 2^{n-\deg(f_{q_k})} < d_{\min}^{{\mathcal I}^{k-1}}$
      and
      $A_{d_{\min}}^{{\mathcal I}^k} = 2^{\deg(f_{q_k}) + \left|\lambda_{f_{q_k}}\right|}$.
  \item If $\deg(f_{q_k}) < r_{\max}^{{\mathcal I}^{k-1}}$, we have
      $d_{\min}^{{\mathcal I}^{k}} = d_{\min}^{{\mathcal I}^{k-1}}$
      and
      $A_{d_{\min}}^{{\mathcal I}^k} = A_{d_{\min}}^{{\mathcal I}^{k-1}}$.
\end{enumerate}

As shown above, the MWD is changed in 1) and 2). The MWD difference is related to $f_{q_k}$.
Thus, the partial MWD is defined as
\begin{definition}\label{DefinitionPartialMWD}
The partial MWD of $W_N^{(i)}$ is a two-tuple $\left(d_i, A_i\right)$ with
\begin{equation}\label{PMWDdmin}
d_i = 2^{n-\deg(f_{i})}
\end{equation}
and
\begin{equation}\label{PMWDAdmin}
A_i = 2^{\deg(f_{i}) + \left|\lambda_{f_{i}}\right|}.
\end{equation}
\end{definition}

The partial MWD $\left(d_i, A_i\right)$ means the influence on the MWD when $W_N^{(i)}$ is selected as the information channel with the information set ${\mathcal A}$ obeying the PO.

\subsection{Construction Sequence Based on MWD}

In this part, we first propose criteria to order the $N$ synthetic channels based on the partial MWD and design a construction sequence. Then, we prove that the sequence is nested and has the optimum performance evaluated by the MWUB in \eqref{union_bound_a} for the polar codes obeying the PO.

The relationships of the partial MWD associated with the PO shown in Definition \ref{def1_PO} and Definition \ref{def2_PO} are provided in Lemma \ref{LemmaPO1} and Lemma \ref{LemmaPO2}, respectively.
\begin{lemma}\label{LemmaPO1}
For the PO in Definition \ref{def1_PO}, given $f_i = \prod_{m=1}^{D}{{{x}_{i_m}}}$ and $f_j = \prod_{m=1}^{D}{{{x}_{j_m}}}$ with $f_i \preceq f_j$ and $D = \deg \left(f_i\right) = \deg \left(f_j\right)$, we have $A_i < A_j$.
\end{lemma}
\begin{proof}
According to \eqref{LamedaCal} and \eqref{PMWDAdmin}, we have
$A_i = 2^{\frac{3D-D^2}{2} + \sum_{m=1}^{D}{i_m}}$ and
$A_j = 2^{\frac{3D-D^2}{2} + \sum_{m=1}^{D}{j_m}}$.
By Definition \ref{def1_PO}, we have
$\sum_{m=1}^{D}{i_m} < \sum_{m=1}^{D}{j_m}$.
Thus, $A_i < A_j$ is derived.
\end{proof}

\begin{lemma}\label{LemmaPO2}
For the PO in Definition \ref{def2_PO}, given two monomials $f_i$ and $f_j $ with $\deg(f_i) < \deg(f_j)$ and $f_i \preceq f_j$, we have $d_i > d_j$.
\end{lemma}
\begin{proof}
The proof is clear by \eqref{PMWDdmin}.
\end{proof}

Then, given $U\left(W_N^{\left(i\right)}\right)$ is the reliability of $W_N^{\left(i\right)}$ with the new PO based on MWD, the order criteria for $W_N^{(i)}$ based on the partial MWD are as follows:
\begin{enumerate}
  \item According to Lemma \ref{LemmaPO2}, $W_N^{(i)}$ with larger $d_i$ has larger $U\left(W_N^{\left(i\right)}\right)$ than others.
  \item According to Lemma \ref{LemmaPO1}, when the synthetic channels have the identical $d_i$, $W_N^{(i)}$ with less $A_i$ has larger $U\left(W_N^{\left(i\right)}\right)$ than the others.
  \item When the synthetic channels have the identical $d_i$ and $A_i$, $W_N^{(i)}$ with larger $i$ has larger $U\left(W_N^{\left(i\right)}\right)$ than others empirically.
\end{enumerate}

Based on criterion 1), the monomials of $W_N^{(i)}$, $i = 0, 1, \cdots, N-1$ are divided into $(n+1)$ subsets
\begin{equation}\label{DivideM}
{\mathcal M}_r = \left\{ f_i : \deg(f_i) = r, f_i \in {\mathcal M} \right\}, r = 0,1,\cdots,n.
\end{equation}
and the monomials in ${\mathcal M}_r$ have the identical $d_i$, i.e.,
\begin{equation}\label{dMonomialSubset}
d_i = 2^{n-r}, f_i \in {\mathcal M}_r.
\end{equation}
Then, let ${\mathcal B}_r$ be the subset of the indices of the synthetic channels corresponding to the monomials in ${\mathcal M}_r$, i.e.,
\begin{equation}\label{DivideA}
{\mathcal B}_r = \left\{ i  : f_i \in {\mathcal M}_r  \right\}, r = 0,1,\cdots,n.
\end{equation}
Hence, the synthetic channels in different subsets have different reliability levels, i.e.,
\begin{equation}\label{cre1dRelia}
\begin{aligned}
U\left(W_N^{\left(i\right)}\right) > U\left(W_N^{\left(j\right)}\right),
\text{ if }
i \in {\mathcal B}_r, j \in {\mathcal B}_{s} \text{ and } r < s.
\end{aligned}
\end{equation}

Then, criterion 2) is used to order the synthetic channels in the identical subset. For $W_N^{\left(i\right)}$ with $i \in {\mathcal B}_r$, calculate $A_i$ by \eqref{PMWDAdmin}. The reliability order is
\begin{equation}
\begin{aligned}
U\left(W_N^{\left(i\right)}\right) > U\left(W_N^{\left(j\right)}\right),
\text{ if }
A_i < A_j, i \in {\mathcal B}_r \text{ and } j \in {\mathcal B}_r.
\end{aligned}
\end{equation}

Finally, since the synthetic channel with a larger index is more reliable empirically when the PO between two synthetic channels is not clear, criterion 3) is designed.
Hence, the synthetic channels in the identical subset with identical partial MWD are ordered as
\begin{equation}
\begin{aligned}
U\left(W_N^{\left(i\right)}\right) > U\left(W_N^{\left(j\right)}\right),
\text{ if }
i > j, A_i = A_j, i \in {\mathcal B}_r \text{ and } j \in {\mathcal B}_r.
\end{aligned}
\end{equation}

\begin{algorithm}[t]
\setlength{\abovecaptionskip}{0.cm}
\setlength{\belowcaptionskip}{-0.cm}
\caption{MWD sequence}\label{SequenceDesign}

\KwIn {The code length $N$;}
\KwOut {The $N$-length MWD sequence ${\bf q}$;}

Calculate $\left(d_i, A_i\right)$ of $W_N^{(i)}, i=0,1,\cdots,N-1$ by \eqref{PMWDdmin} and \eqref{PMWDAdmin}\;

Divide the $N$ synthetic channels into $(n+1)$ subsets ${\mathcal B}_r, r = 0,1,\cdots,n$ by \eqref{DivideA}\;

Initialize $B_r \leftarrow \sum_{i=0}^{r}{\left|{\mathcal B}_i\right|}, r = 0,1,\cdots,n$ and $B_{-1} = 0$\;

\For{$r = 0,1,\cdots,n$}
{
    Initialize a $\left| {\mathcal B}_r \right|$-length sequence ${\bf p}$\;

    Initialize ${\mathcal T} \leftarrow {\mathcal B}_r$ and $j \leftarrow 0$\;

    \While{${\mathcal T} \ne \varnothing$}
    {
        ${\mathcal Q} \leftarrow \underset{i \in {\mathcal T}}{\mathop{\arg \min }}\,A_i $\;
        $a \leftarrow {\mathop{\max }}\, {\mathcal Q}$\;
        $p_j \leftarrow a$,
        $j \leftarrow j + 1$ and
        ${\mathcal T} \leftarrow {\mathcal T} \setminus \left\{a\right\}$\;
    }
    \For{$k = 0,1,\cdots,\left|{\mathcal B}_r\right|-1$}
    {
        ${q}_{B_{r-1}+k} \leftarrow {p}_k$\;
    }
}
\Return $\bf q$\;
\end{algorithm}

According to the criteria, the reliability order is clear and the construction sequence based on MWD is designed directly. The design method is described in Algorithm \ref{SequenceDesign}.

Then, in Lemma \ref{LemmaML}, we prove the polar codes obeying the PO with the MWD sequence have the optimum performance evaluated by the MWUB in the high SNR region.
In Lemma \ref{LemmaNested}, we prove the MWD sequence is nested, which means that the MWD sequence can be used similarly to the polar sequence in 5G \cite{3GPP_5G_polar}.

\begin{lemma}\label{LemmaML}
For all the $(N,K)$ polar codes obeying the PO, the polar code constructed by the MWD sequence $\bf q$ has the optimum MWUB as SNR goes to infinity.
\end{lemma}
\begin{proof}
As SNR goes to infinity, according to \eqref{union_bound_a}, $d_{\min}$ is the key factor and the MWUB increases as $d_{\min}$ decreases. When $d_{\min}$ is fixed, the MWUB increases as $A_{d_{\min}}$ increases. Hence, the optimum polar code evaluated by the MWUB should ensure 1) the maximum $d_{\min}$ and 2) the minimum $A_{d_{\min}}$ when $d_{\min}$ is identical.

Let an $(N,K)$ polar code ${\mathcal C}$ constructed by $\bf q$ has the information set
\begin{equation}\label{MWDSA}
{\mathcal A} = \left\{q_0,q_1,\cdots,q_{K-1} \right\}.
\end{equation}
The information set ${\mathcal A}'$ of another polar code ${\mathcal C}'$ obeying the PO has at least one element with the index equal or larger than $K$, i.e.,
$\exists q_i \in {\mathcal A}'$, $i \ge K$.
According to criterion 1), the minimum Hamming weight of ${\mathcal C}$ is $d_{\min} = d_{q_{K-1}}$ and we have
$d_{q_{K-1}} \ge d_i, i = K,K+1,\cdots,N-1$.
Thus, the polar code ${\mathcal C}$ constructed by $\bf q$ has the maximum $d_{\min}$.

Assuming that the minimum Hamming weight of ${\mathcal C}'$ is
$d_{\min}' = d_{\min}$, the numbers of codewords of ${\mathcal C}$ and ${\mathcal C}'$ with the minimum Hamming weight $d_{\min}$ are
$A_{d_{\min}} = \sum_{i \in ({\mathcal A} \cap {\mathcal B}_{r}) }{A_i}$ and
$A_{d_{\min}'} = \sum_{i \in ({\mathcal A}' \cap {\mathcal B}_{r}) }{A_i}$, respectively, where $r = n - \log{d_{\min}}$. According to criterion 2), $\forall \overline{\mathcal B} \subseteq {\mathcal B}_{r}$ with $|\overline{\mathcal B}| = |{\mathcal A} \cap {\mathcal B}_{r}|$ makes
\begin{equation}\label{EqLeAdmin1}
A_{d_{\min}} \le \sum_{i \in \overline{\mathcal B}}{A_i}.
\end{equation}
Then, since
$|{\mathcal A} \cap {\mathcal B}_{r}| \le |{\mathcal A}' \cap {\mathcal B}_{r}|$, we can obtain
$\exists \overline{\mathcal B} \subseteq ({\mathcal A}' \cap {\mathcal B}_{r})$ with $|\overline{\mathcal B}| = |{\mathcal A} \cap {\mathcal B}_{r}|$ making
\begin{equation}\label{EqLeAdmin2}
\sum_{i \in \overline{\mathcal B}}{A_i} \le A_{d_{\min}'}.
\end{equation}
According to \eqref{EqLeAdmin1} and \eqref{EqLeAdmin2}, we have
$A_{d_{\min}} \le A_{d_{\min}'}$.
Hence, the polar code constructed by $\bf q$ has the minimum $A_{d_{\min}}$ when $d_{\min}$ is identical and the polar code obeys the PO. Thus, Lemma \ref{LemmaML} is proved.
\end{proof}

\begin{figure*}[t]
\setlength{\abovecaptionskip}{0.cm}
\setlength{\belowcaptionskip}{-0.cm}
  \centering{\includegraphics[scale=0.67]{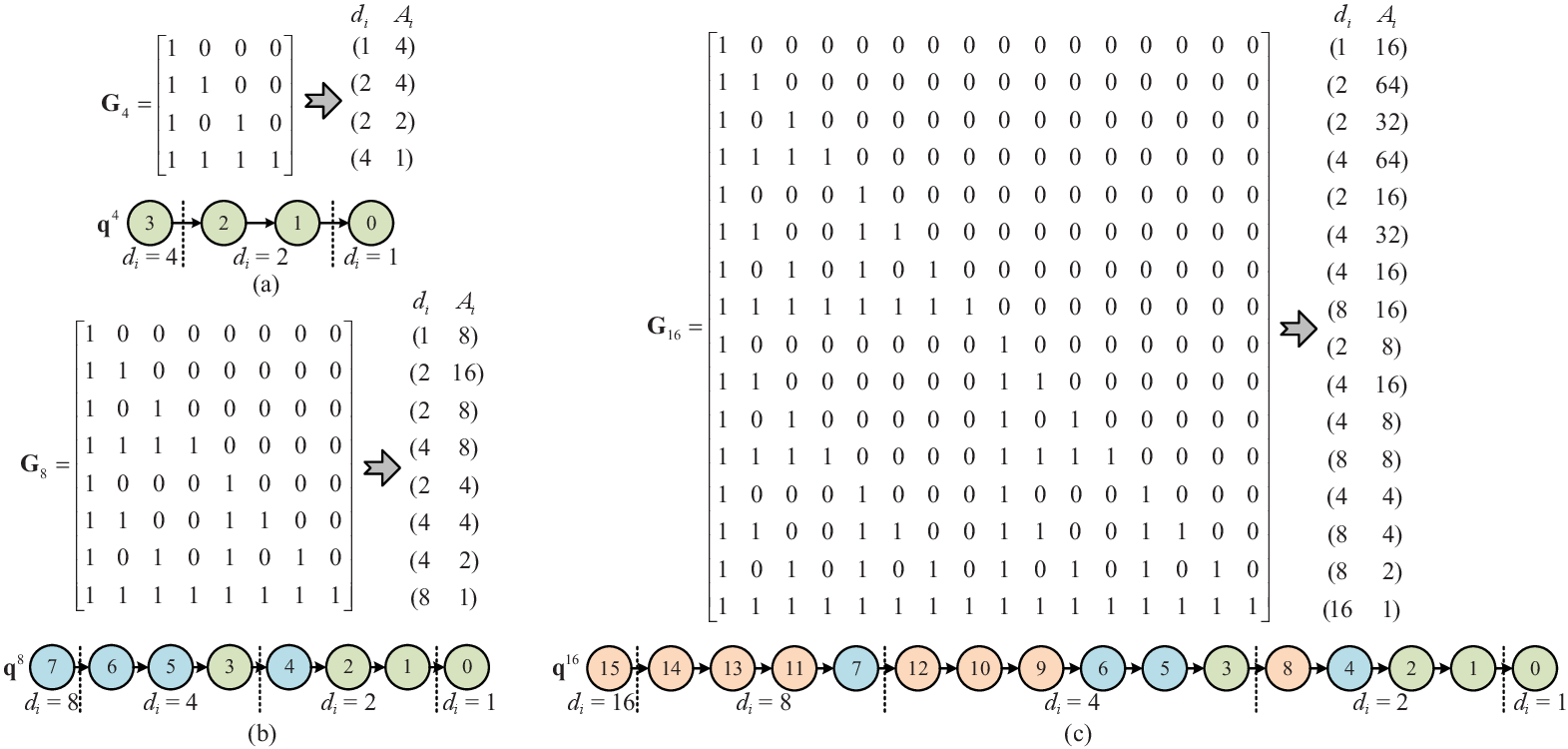}}
  \caption{Fig. \ref{FigExampleMWD}(a), Fig. \ref{FigExampleMWD}(b) and Fig. \ref{FigExampleMWD}(c) illustrate the examples of the nested MWD sequences with $N = 4, 8, 16$, respectively.}\label{FigExampleMWD}
  \vspace{0em}
\end{figure*}

\begin{lemma}\label{LemmaNested}
Let ${\bf q}^N = \left[q_0^{N}, q_1^{N},\cdots,q_{N-1}^{N}\right]$ and ${\bf q}^{2N} = \left[q_0^{2N}, q_1^{2N},\cdots,q_{2N-1}^{2N}\right]$ be $N$-length and $2N$-length MWD sequences, respectively. For any ${q}_i^{N}$ and ${q}_j^{N}$ with $i < j$, we can find ${q}_k^{2N} = q_i^{N}$ and ${q}_l^{2N} = q_j^{N}$ with $k < l$, which means that the MWD sequence is nested.
\end{lemma}
\begin{proof}
Let $f_{{q}_i^{N}}$, $f_{{q}_j^{N}}$, $f_{{q}_k^{2N}}$ and $f_{{q}_l^{2N}}$ be the monomials corresponding to ${q}_i^{N}$, ${q}_j^{N}$, ${q}_k^{2N}$ and ${q}_l^{2N}$,  respectively.
Given ${q}_k^{2N} = q_i^{N}$ and ${q}_l^{2N} = q_j^{N}$, we have
\begin{equation}
\begin{aligned}
\label{LemmaNestEq1}
{q}_k^{2N} = q_i^{N} \Rightarrow f_{{q}_k^{2N}} = f_{{q}_i^{N}} \cdot x_{n} \Rightarrow \deg(f_{{q}_k^{2N}}) = \deg(f_{{q}_i^{N}}) + 1,
\end{aligned}
\end{equation}
\begin{equation}
\begin{aligned}
\label{LemmaNestEq2}
{q}_l^{2N} = q_j^{N} \Rightarrow f_{{q}_l^{2N}} = f_{{q}_j^{N}} \cdot x_{n} \Rightarrow \deg(f_{{q}_l^{2N}}) = \deg(f_{{q}_j^{N}}) + 1.
\end{aligned}
\end{equation}
Assuming $i < j$, we have $d_{{q}_i^{N}} \ge d_{{q}_j^{N}}$, which leads to
\begin{equation}
\deg(f_{{q}_i^{N}}) \le \deg(f_{{q}_j^{N}}) \Rightarrow \deg(f_{{q}_k^{2N}}) \le \deg(f_{{q}_l^{2N}}).
\end{equation}
If $\deg(f_{{q}_i^{N}}) < \deg(f_{{q}_j^{N}})$, by \eqref{LemmaNestEq1} and \eqref{LemmaNestEq2}, we obtain
\begin{equation}
\deg(f_{{q}_k^{2N}}) < \deg(f_{{q}_l^{2N}}) \Rightarrow d_{{q}_k^{2N}} > d_{{q}_l^{2N}} \Rightarrow k < l.
\end{equation}
If $\deg(f_{{q}_i^{N}}) = \deg(f_{{q}_j^{N}})$, we have
$A_{{{q}_i^{N}}} \le A_{{{q}_j^{N}}}$, since $i < j$.
By \eqref{PMWDAdmin}, we have
\begin{equation}\label{LemmaNestEq5}
A_{{{q}_k^{2N}}} \le A_{{{q}_l^{2N}}}.
\end{equation}
Equality holds in \eqref{LemmaNestEq5} iff $A_{{{q}_i^{N}}} = A_{{{q}_j^{N}}}$.
For $A_{{{q}_k^{2N}}} < A_{{{q}_l^{2N}}}$,we have $k < l$ by criterion 2).
For $A_{{{q}_k^{2N}}} = A_{{{q}_l^{2N}}}$, we have $A_{{{q}_i^{N}}} = A_{{{q}_j^{N}}}$, which derives
\begin{equation}
{q}_i^{N} > {q}_j^{N} \Rightarrow {q}_k^{2N} > {q}_l^{2N} \Rightarrow k < l
\end{equation}
by criterion 3).
Hence, Lemma \ref{LemmaNested} is proved.
\end{proof}

Then, Fig. \ref{FigExampleMWD} provides examples to illustrate the MWD sequence. In Fig. \ref{FigExampleMWD} (a), Fig. \ref{FigExampleMWD}(b) and Fig. \ref{FigExampleMWD}(c), the examples of the MWD sequences with $N = 4, 8, 16$ are provided, respectively.
For each example, we calculate the partial MWD by \eqref{PMWDdmin} and \eqref{PMWDAdmin} and obtain the MWD sequence by Algorithm \ref{SequenceDesign}.
The MWD sequences in Fig. \ref{FigExampleMWD}(a), Fig. \ref{FigExampleMWD}(b) and Fig. \ref{FigExampleMWD}(c) are
\begin{align}
{\bf q}^4 &= \left[3,2,1,0\right],\\
{\bf q}^8 &= \left[7,6,5,3,4,2,1,0\right],\\
{\bf q}^{16} &= \left[15,14,13,11,7,12,10,9,6,5,3,8,4,2,1,0\right].
\end{align}
The elements colored green in ${\bf q}^4$ have the identical order in ${\bf q}^8$ and ${\bf q}^{16}$. Then, the elements colored blue in ${\bf q}^8$ also have the same order in ${\bf q}^{16}$. Thus, ${\bf q}^4$ is nested in ${\bf q}^8$ which is also nested in ${\bf q}^{16}$.

\begin{remark}
By Lemma \ref{LemmaNested}, given an $N_{\max}$-length MWD sequence ${\bf q}^{N_{\max}}$, any $(N,K)$ polar code with $N \le N_{\max}$ can be constructed by ${\bf q}^{N_{\max}}$. The construction method is identical to the polar sequence in 5G, which means the MWD sequence is practical.
\end{remark}

\begin{remark}
Though polar codes obeying PO with the MWD sequence have the optimum MWUB in the high SNR region by Lemma \ref{LemmaML}, achieving the ML performance is difficult for long code length under the SCL decoding with limited list size. Thus, a suitable construction method based on MWD for SCL decoding is required.
\end{remark}

\section{Construction Methods Based on MWD for SCL Decoding}

In this section, we first introduce the entropy constraint according to the information-theoretic perspective on SCL decoding. Then, we use the entropy constraint to provide a construction method called entropy constraint bit-swapping (ECBS) algorithm.

\subsection{Entropy Constraint}

The information-theoretic perspective on SCL decoding proposed in \cite{InformationSCL} establishes the relationship between the average list size and the ML decoding over binary memoryless symmetric channels in Theorem \ref{theorem1}.

\begin{theorem}[\unskip\cite{InformationSCL}]\label{theorem1}
The set of survival path ${\hat {\bf u}}_0^i$ with a larger likelihood than some fraction of that for true path ${\bf u}_0^i$ after the $i$-stage of SCL decoding is
\begin{equation}\label{theoEq1}
{\mathcal S}_{\alpha}^{(i)}\left({\bf u}_0^i, {\bf y}\right) \triangleq \left\{{\widehat {\bf u}}_0^i:P\left({\widehat {\bf u}}_0^i | {\bf y} \right) \ge \alpha P\left({{\bf u}}_0^i | {\bf y} \right) \right\},
\end{equation}
where the fraction is $\alpha \le 1$. The average of the logarithm of the cardinality is upper bounded by
\begin{equation}\label{theo1upper}
\mathbb{E}\left[ \log{|{\mathcal S}_{\alpha}^{(i)}|} \right] \le \log{\alpha^{-1}} +\sum_{j \in {\mathcal A}^{(i)}}{H\left(W_N^{(j)}\right)},
\end{equation}
where ${\mathcal A}^{(i)} \triangleq {\mathcal A} \bigcap \left\{0,1,\cdots,i\right\}$ and ${H\left(W_N^{(j)}\right)}$ is the entropy of $W_N^{(j)}$.
\end{theorem}

Assuming $\alpha = 1$, \eqref{theoEq1} is simplified as
\begin{equation}
{\mathcal S}_{1}^{(i)}\left({\bf u}_0^i, {\bf y}\right) =
\left\{{\widehat {\bf u}}_0^i:P\left({\widehat {\bf u}}_0^i | {\bf y} \right) \ge P\left({{\bf u}}_0^i | {\bf y} \right) \right\}.
\end{equation}
Given the information sequence ${\bf u}_0^{N-1}$, the SCL decoding can achieve the ML performance if $L$ satisfies
\begin{equation}\label{BGRcre1}
L \ge \underset{\mathbf{y}\in {{\mathsf{\mathcal{Y}}}^{N}, i = 0,1,\cdots,N-1}}{\mathop{\max }}\,|\mathsf{\mathcal{S}}_{1}^{(i)}|.
\end{equation}

Due to obtaining the RHS of \eqref{BGRcre1} requiring high search complexity, we propose a heuristic entropy constraint at the $i$-th stage by using the upper bound of the average list size in \eqref{theo1upper} to evaluate if $L$ is large enough to achieve the ML performance, i.e.,
\begin{equation}\label{BGRcre3}
\log{L} \ge \sum_{j \in {\mathcal A}^{(i)}}{H\left(W_N^{(j)}\right)} \ge \mathbb{E}\left[ \log{|{\mathcal S}_{1}^{(i)}|} \right], i = 0,1,\cdots,N-1.
\end{equation}
Then, to make the SCL decoding after the $N$-stage satisfying \eqref{BGRcre3}, $L$ should satisfy the heuristic entropy constraint as
\begin{equation}\label{BGRcre4}
\log{L} \ge {H\left({\mathcal A}\right)},
\end{equation}
where ${H\left({\mathcal A}\right)} \triangleq \sum_{j \in {\mathcal A}}{H_j}$ and $H_j \triangleq H\left(W_N^{(j)}\right)$.

\subsection{Entropy Constraint Bit-Swapping Algorithm}

\begin{algorithm}[t]
\setlength{\abovecaptionskip}{0.cm}
\setlength{\belowcaptionskip}{-0.cm}
\caption{ECBS algorithm}\label{BGRMWDGA}

\KwIn {The code length $N$, the information length $K$, the list size $L$ and the SNR $\frac{E_b}{N_0}$;}
\KwOut {The information set ${\mathcal A}$;}

Divide the $N$ synthetic channels into $(n+1)$ subsets ${\mathcal B}_r, r = 0,1,\cdots,n$ by \eqref{DivideA}\;
Calculate $H_j$, $j=0,1,\cdots,N-1$\;

\For{$b \leftarrow 0,1,\cdots,n$}
{
    Initialize ${\mathcal A}$ by the MWD sequence ${\bf q}$\;
    Set ${\mathcal B} \leftarrow \mathop  \bigcup_{r = 0}^b {{\cal B}_r}$\;
    \If{$\left|{\mathcal B}\right| \ge K$}
    {
        \While{$\log{L} < {H\left({\mathcal A}\right)}$}
        {
            ${\mathcal F}' \leftarrow {\mathcal B} \setminus {\mathcal A}$\;
            $j_{\max} \leftarrow \mathop{\arg\max}\limits_{j \in {\mathcal A}}{H_j}$ and $H_{\max} \leftarrow H_{j_{\max}}$\;
            $j_{\min} \leftarrow \mathop{\arg \min}\limits_{j \in {\mathcal F}'} {H_j}$ and $H_{\min} \leftarrow H_{j_{\min}}$\;
            \If{$H_{\max} > H_{\min}$}
            {
                ${\mathcal A} \leftarrow {\mathcal A} \bigcup \left\{j_{\min}\right\}\setminus\left\{j_{\max}\right\}$\;
            }
            \Else
            {
                {\bf break}\;
            }
        }
        \If{$\log L \ge H\left({\mathcal A}\right)$ or $b = n$}
        {
            \Return ${\mathcal A}$\;
        }

    }
}

\end{algorithm}

To satisfy the entropy constraint \eqref{BGRcre4}, we propose an ECBS algorithm. Given an $\left(N, K\right)$ polar code, the ECBS algorithm first initializes the information set $\mathcal A$ by the MWD sequence $\bf q$. Then, we divide the synthetic channels into $(n+1)$ groups by the partial MWD and gradually swap the information bit and frozen bit to reduce ${H\left({\mathcal A}\right)}$ until the entropy constraint is satisfied.

\begin{figure*}[t]
\setlength{\abovecaptionskip}{0.cm}
\setlength{\belowcaptionskip}{-0.cm}
  \centering{\includegraphics[scale=0.62]{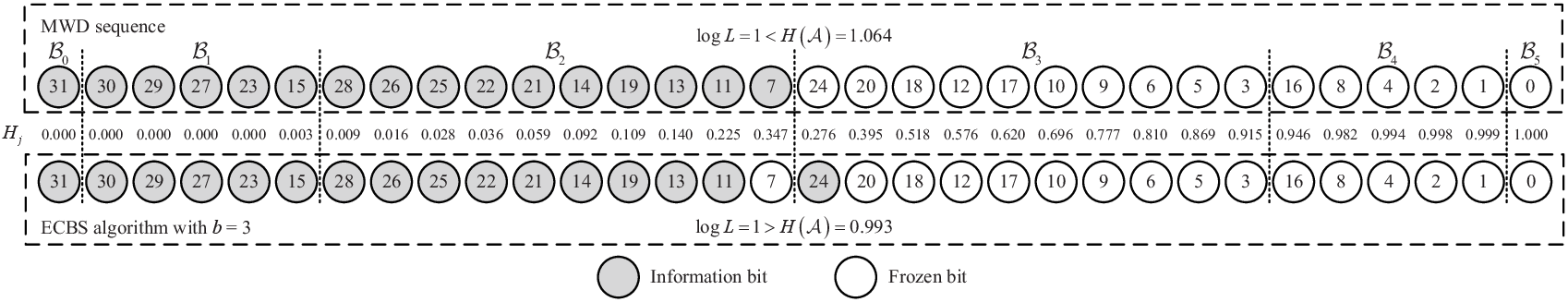}}
  \caption{The information set obtained by ECBS algorithm with $N=32$, $K=16$, $L=2$ and $E_b/N_0=1.25$dB.}\label{FigExampleBGRMWDExample}
  \vspace{0em}
\end{figure*}

In Algorithm \ref{BGRMWDGA}, we first divide $N$ synthetic channels into $(n+1)$ subsets ${\mathcal B}_r, r = 0,1,\cdots,n$ by \eqref{DivideA} and determine the index set ${\mathcal B} = \mathop  \bigcup_{r = 0}^b {{\cal B}_r}$. ${\mathcal B}$ decides the swapping range and $d_{\min}$.
In the swapping process, we gradually increase $b$ to enlarge ${\mathcal B}$ and try not to reduce $d_{\min}$.
Then, we initialize the information set ${\mathcal A}$ by the MWD sequence ${\bf q}$ and determine the swapping frozen set ${\mathcal F}' = {\mathcal B} \setminus {\mathcal A}$.
In ${\mathcal A}$ and ${\mathcal F}'$, we find the information bit with the maximum entropy, i.e., $j_{\max} = {\arg\max}_{j \in {\mathcal A}}{H_j}$ and $H_{\max} = H_{j_{\max}}$, and the frozen bit with the minimum entropy, i.e., $j_{\min} = {\arg \min}_{j \in {\mathcal F}'} {H_j}$ and $H_{\min} = H_{j_{\min}}$, respectively.
If $H_{\max} > H_{\min}$, we swap the information bit and the frozen bit and update the information set as ${\mathcal A} \leftarrow {\mathcal A} \bigcup \left\{j_{\min}\right\}\setminus\left\{j_{\max}\right\}$.
After each bit-swapping, we judge $\mathcal A$ satisfying the entropy constraint $\log{L} \ge {H\left({\mathcal A}\right)}$ or not.
If satisfying, $\mathcal A$ is the information set obtained by the ECBS algorithm. If not, we continue to swap the bits in ${\mathcal A}$ and ${\mathcal F}'$.
When swapping cannot reduce ${H\left({\mathcal A}\right)}$, i.e., $H_{\max} \le H_{\min}$, we update $b \leftarrow b+1$ and repeat the process above until $\log{L} \ge {H\left({\mathcal A}\right)}$ or $b = n$.

{Fig. \ref{FigExampleBGRMWDExample} shows the information sets obtained by the MWD sequence and ECBS algorithm, where $N=32$, $K=16$, $L=2$ and $E_b/N_0=1.25$dB.
In Fig. \ref{FigExampleBGRMWDExample}, the information set ${\mathcal A}$ decided by ${\bf q}$ is
\begin{equation}\label{EqExampleBGR1}
{\mathcal A} = \left\{31,30,29,27,23,15,28,26,25,22,21,14,19,13,11,7\right\}.
\end{equation}
Then, due to
\begin{equation}
\log L=1 < {H\left({\mathcal A}\right)}=1.064,
\end{equation}
we should optimize ${\mathcal A}$ to satisfy \eqref{BGRcre4}.

By Algorithm \ref{BGRMWDGA} with $b = 3$, ${\mathcal A}$ is modified as
\begin{equation}\label{EqExampleBGR2}
{\mathcal A} = \left\{31,30,29,27,23,15,28,26,25,22,21,14,19,13,11,24\right\},
\end{equation}
which satisfies \eqref{BGRcre4}, i.e.,
\begin{equation}
\log L=1 > {H\left({\mathcal A}\right)}=0.993.
\end{equation}

\begin{table}[t]
  \centering
  \caption{The MWD sequence ${\bf q}^{256}$ with $N = 256$.}\label{TabMWDsequenceN256}
  {
\begin{tabular}{|c|c|c|c|c|c|c|c|}
\hline
${\bf q}_0^{31}$& ${\bf q}_{32}^{63}$ & ${\bf q}_{64}^{95}$  & ${\bf q}_{96}^{127}$ & ${\bf q}_{128}^{159}$& ${\bf q}_{160}^{191}$ & ${\bf q}_{192}^{223}$ & ${\bf q}_{224}^{255}$\\ \hline
255 & 119 & 118 & 216 & 105 & 27  & 134 & 80  \\ \hline
254 & 159 & 203 & 226 & 102 & 23  & 81  & 132 \\ \hline
253 & 111 & 179 & 212 & 90  & 15  & 74  & 72  \\ \hline
251 & 95  & 173 & 184 & 60  & 224 & 50  & 48  \\ \hline
247 & 63  & 158 & 225 & 163 & 208 & 44  & 130 \\ \hline
239 & 248 & 117 & 210 & 149 & 200 & 133 & 68  \\  \hline
223 & 244 & 110 & 204 & 142 & 176 & 73  & 40  \\  \hline
191 & 242 & 199 & 180 & 101 & 196 & 70  & 129 \\  \hline
127 & 236 & 171 & 120 & 89  & 168 & 49  & 66  \\  \hline
252 & 241 & 157 & 209 & 86  & 112 & 42  & 36  \\  \hline
250 & 234 & 115 & 202 & 58  & 194 & 28  & 24  \\  \hline
249 & 220 & 109 & 178 & 147 & 164 & 131 & 65  \\  \hline
246 & 233 & 94  & 172 & 141 & 152 & 69  & 34  \\  \hline
245 & 230 & 167 & 116 & 99  & 104 & 41  & 20  \\  \hline
238 & 218 & 155 & 201 & 85  & 193 & 38  & 33  \\  \hline
243 & 188 & 107 & 198 & 78  & 162 & 26  & 18  \\  \hline
237 & 229 & 93  & 177 & 57  & 148 & 67  & 12  \\  \hline
222 & 217 & 62  & 170 & 54  & 100 & 37  & 17  \\  \hline
235 & 214 & 151 & 156 & 139 & 88  & 25  & 10  \\  \hline
221 & 186 & 103 & 114 & 83  & 161 & 22  & 9   \\  \hline
190 & 124 & 91  & 108 & 77  & 146 & 35  & 6   \\  \hline
231 & 227 & 61  & 197 & 53  & 140 & 21  & 5   \\  \hline
219 & 213 & 143 & 169 & 46  & 98  & 14  & 3   \\  \hline
189 & 206 & 87  & 166 & 135 & 84  & 19  & 128 \\  \hline
126 & 185 & 59  & 154 & 75  & 56  & 13  & 64  \\  \hline
215 & 182 & 79  & 113 & 51  & 145 & 11  & 32  \\  \hline
187 & 122 & 55  & 106 & 45  & 138 & 7   & 16  \\  \hline
125 & 211 & 47  & 92  & 30  & 97  & 192 & 8   \\  \hline
207 & 205 & 31  & 195 & 71  & 82  & 160 & 4   \\  \hline
183 & 181 & 240 & 165 & 43  & 76  & 144 & 2   \\  \hline
123 & 174 & 232 & 153 & 29  & 52  & 96  & 1   \\  \hline
175 & 121 & 228 & 150 & 39  & 137 & 136 & 0   \\  \hline
\end{tabular}
}
\vspace{0em}
\end{table}

\section{Performance Evaluation}

In this section, we first provide the performance of polar codes with MWD sequence. Then, the performance of polar codes constructed by the ECBS algorithm is illustrated.

\subsection{Simulation Results of MWD Sequence}

In this subsection, the MWD sequence obtained by Algorithm \ref{SequenceDesign} is first provided in Table \ref{TabMWDsequenceN256}. Then, the BLER performance of polar codes constructed by the MWD sequence is shown. The polar codes are decoded by SCL decoding.

\begin{figure}[t]
\setlength{\abovecaptionskip}{0.cm}
\setlength{\belowcaptionskip}{-0.cm}
  \centering{\includegraphics[scale=0.69]{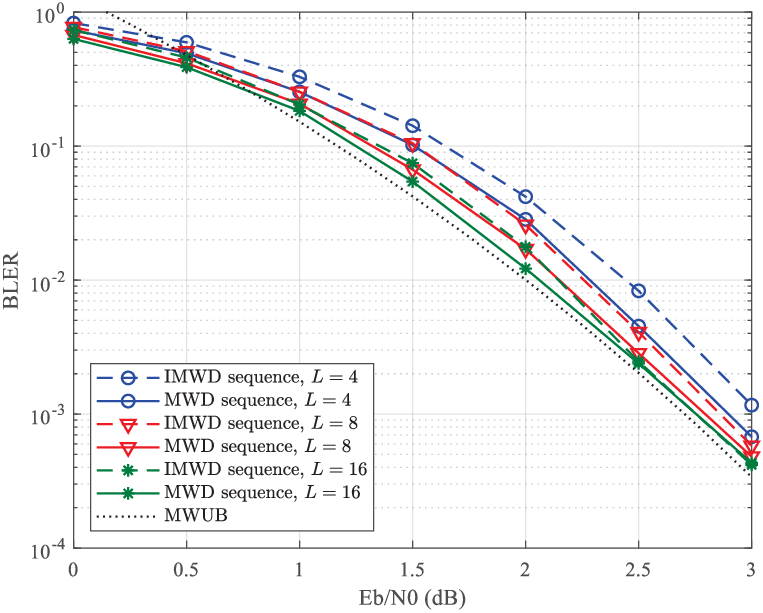}}
  \caption{The BLER performance of polar codes with $N = 256$ and $K = 128$, where IMWD sequence means the inverse MWD sequence with the opposite criterion 3).}\label{FigN256MWDinv}
  \vspace{0em}
\end{figure}

\begin{figure}[t]
\setlength{\abovecaptionskip}{0.cm}
\setlength{\belowcaptionskip}{-0.cm}
  \centering{\includegraphics[scale=0.69]{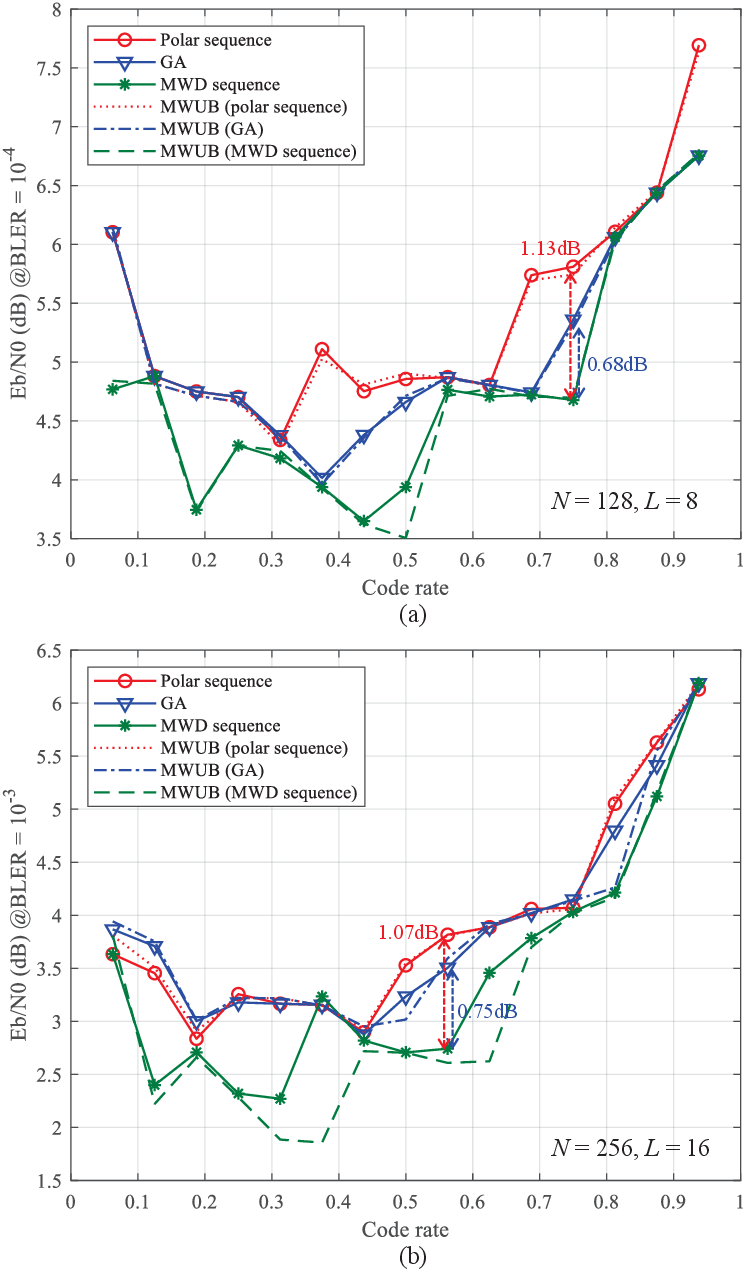}}
  \caption{Fig. \ref{FigTargetSNR}(a) and Fig. \ref{FigTargetSNR}(b) illustrate the minimum required SNRs of polar codes decoded by SCL decoding with $L = 8$ and $L = 16$ to achieve BLER $\le 10^{-4}$ and BLER $\le 10^{-3}$ under the AWGN channel with $N = 128$ and $N = 256$, respectively.}\label{FigTargetSNR}
  \vspace{0em}
\end{figure}

Fig. \ref{FigN256MWDinv} provides the BLER performance of polar codes with $N=256$ and $K=128$, where IMWD sequence means the inverse MWD sequence with the opposite criterion 3), i.e., when the synthetic channels have the identical $d_i$ and $A_i$, $W_N^{(i)}$ with less $i$ has larger $U\left(W_N^{\left(i\right)}\right)$ than others.
The MWD of polar codes with the two construction methods is identical, i.e., $d_{\min} = 16$ and $A_{d_{\min}} = 42288$.
We can observe that the performance of IMWD sequence is worse than that of MWD sequence for $L=4,8$ and $16$ with the performance gaps $0.17$dB, $0.11$dB and $0.08$dB at BLER $10^{-2}$, respectively.
Then, for $L = 16$, the performance of the two sequences is close to the MWUB in the high SNR region.
Thus, polar codes with the empirically designed criterion 3) have better performance than the opposite criterion 3).

Fig. \ref{FigTargetSNR} shows the minimum required SNRs of polar codes constructed by polar sequence \cite{3GPP_5G_polar}, GA algorithm \cite{GA} and MWD sequence to achieve the target BLER under the AWGN channel with the code rate range $R = 0.0625 \sim 0.9375$. The required SNRs of MWUB equal to the target BLER are also provided.
The MWD sequence in Table \ref{TabMWDsequenceN256} is used to construct polar codes.

Fig. \ref{FigTargetSNR}(a) provides the required SNRs of SCL decoding with $N=128$, $L = 8$ and BLER $\le 10^{-4}$.
In Fig. \ref{FigTargetSNR}(a), we observe that the required SNRs of polar sequence \cite{3GPP_5G_polar}, GA algorithm \cite{GA} and MWD sequence are close to the required SNRs of corresponding MWUB.
Then, since the polar codes constructed by the MWD sequence have the optimum MWUB, the required SNRs are less than or equal to those of the polar sequence and GA algorithm.
Specifically, the MWD sequence shows about $0.68$dB and $1.13$dB SNR gaps at $R = 0.75$ compared with the GA algorithm and the polar sequence, respectively.
Thus, MWD sequence can improve the performance of polar codes with medium list size for short code length.
Next, there is about $0.45$dB SNR gap between the MWD sequence and the corresponding MWUB at $R = 0.5$, which means SCL decoding with list size larger than $8$ can provide less minimum required SNR to achieve BLER $\le 10^{-4}$.
Fig. \ref{FigTargetSNR}(b) provides the required SNRs of SCL decoding with $N=256$, $L = 16$ and BLER $\le 10^{-3}$.
In Fig. \ref{FigTargetSNR}(b), the optimum MWUB of the MWD sequence also leads to the corresponding minimum required SNRs less than or equal to those of GA algorithm and polar sequence.
The MWD sequence has about $0.75$dB and $1.07$dB SNR gaps at $R = 0.5625$ compared with the GA algorithm and the polar sequence, respectively.
Compared with Fig. \ref{FigTargetSNR}(a), more cases show SNR gaps between the MWD sequence and the MWUB.
Specifically, the SNR gaps are $0.38$dB, $1.37$dB and $0.83$dB at $R = 0.3125$, $R = 0.375$ and $R = 0.625$, respectively.
The reason is that as the code length increases, the polar codes constructed by the MWD sequence cannot satisfy \eqref{BGRcre4}, which means the required SNRs of the MWD sequence with limited list size deviate from the ML performance. Thus, increasing list size can further improve the performance of MWD sequence to approach the ML performance.

\begin{figure}[t]
\setlength{\abovecaptionskip}{0.cm}
\setlength{\belowcaptionskip}{-0.cm}
  \centering{\includegraphics[scale=0.68]{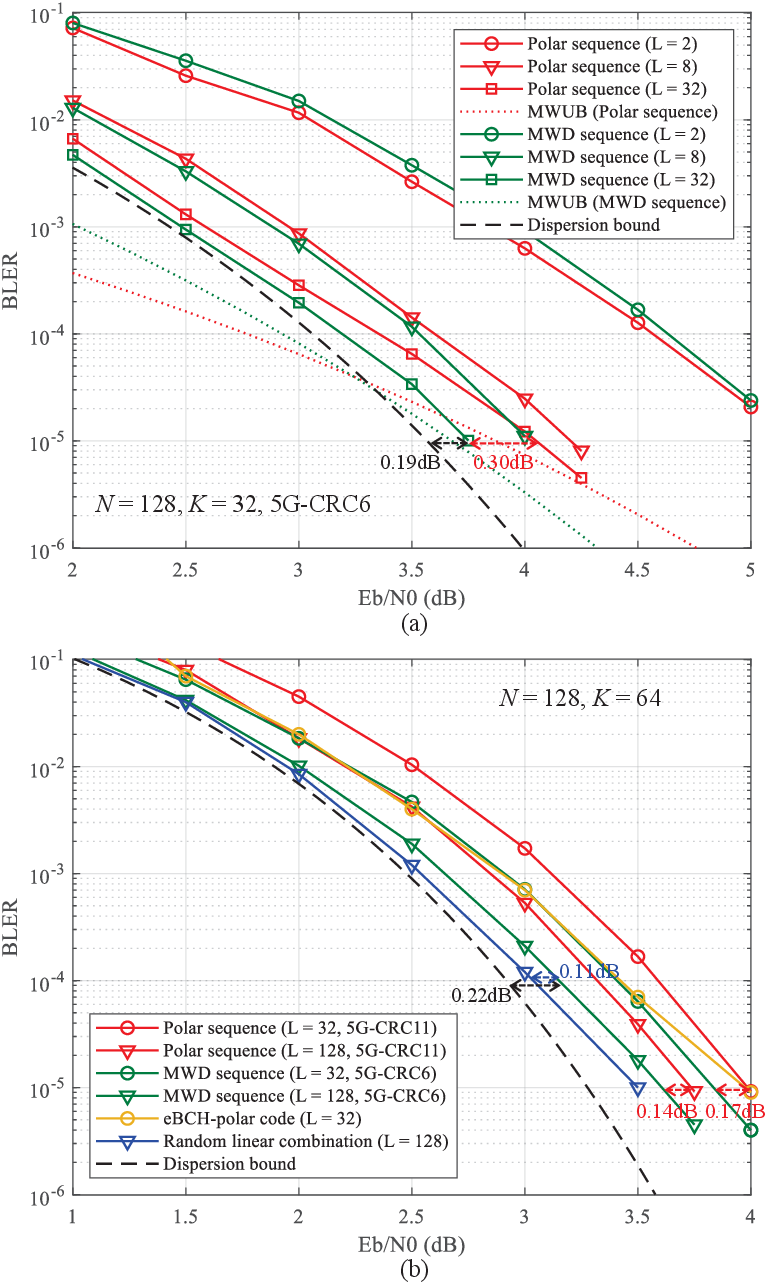}}
  \caption{Fig. \ref{FigN128CRC6}(a) and Fig. \ref{FigN128CRC6}(b) show the BLER performance of $(128,32)$ and $(128,64)$ CRC-polar codes, respectively.}\label{FigN128CRC6}
  \vspace{0em}
\end{figure}

\begin{figure}[t]
\setlength{\abovecaptionskip}{0.cm}
\setlength{\belowcaptionskip}{-0.cm}
  \centering{\includegraphics[scale=0.68]{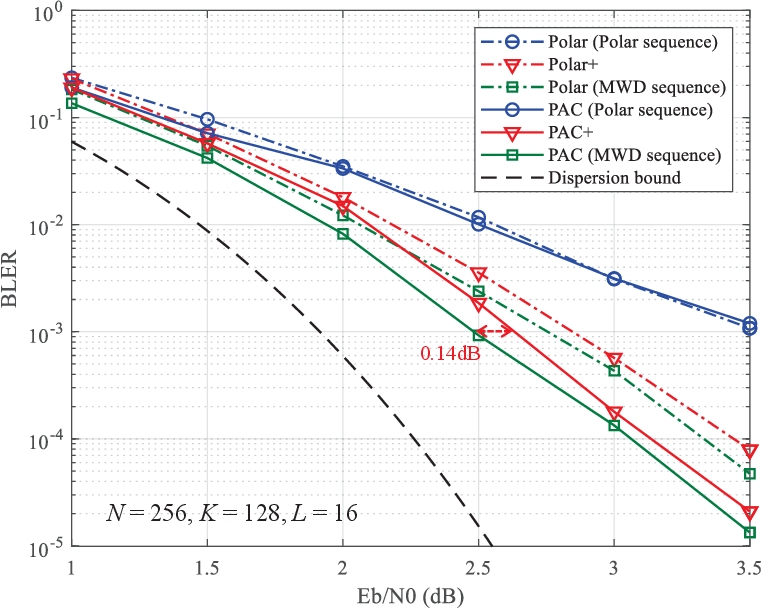}}
  \caption{The BLER performance of polar codes and PAC codes constructed by different methods with $N = 256$, $K = 128$ and $L = 16$.}\label{FigN256}
  \vspace{-0em}
\end{figure}

\begin{table}[t]
\centering
  \caption{The MWDs of CRC-polar codes constructed by polar sequence and MWD sequence with $N = 128$ and 5G-CRC6.}\label{TabMWDN128CRC6}
\begin{tabular}{|c|c|cc|cc|}
\hline
    &     & \multicolumn{2}{c|}{Polar sequence} & \multicolumn{2}{c|}{MWD sequence} \\ \hline
$N$   & $K$   & \multicolumn{1}{c|}{$d_{\min}$}  & $A_{d_{\min}}$  & \multicolumn{1}{c|}{$d_{\min}$} & $A_{d_{\min}}$ \\ \hline
128 & 8   & \multicolumn{1}{c|}{48}    & 6      & \multicolumn{1}{c|}{48}   & 6     \\ \hline
128 & 16  & \multicolumn{1}{c|}{32}    & 20     & \multicolumn{1}{c|}{32}   & 14    \\ \hline
128 & 24  & \multicolumn{1}{c|}{24}    & 4      & \multicolumn{1}{c|}{32}   & 151   \\ \hline
128 & 32  & \multicolumn{1}{c|}{16}    & 2      & \multicolumn{1}{c|}{24}   & 164   \\ \hline
128 & 40  & \multicolumn{1}{c|}{16}    & 46     & \multicolumn{1}{c|}{16}   & 46    \\ \hline
128 & 48  & \multicolumn{1}{c|}{16}    & 293    & \multicolumn{1}{c|}{16}   & 165   \\ \hline
128 & 56  & \multicolumn{1}{c|}{12}    & 8      & \multicolumn{1}{c|}{16}   & 706   \\ \hline
128 & 64  & \multicolumn{1}{c|}{8}     & 19     & \multicolumn{1}{c|}{12}   & 34    \\ \hline
128 & 72  & \multicolumn{1}{c|}{8}     & 51     & \multicolumn{1}{c|}{8}    & 56    \\ \hline
128 & 80  & \multicolumn{1}{c|}{8}     & 199    & \multicolumn{1}{c|}{8}    & 199   \\ \hline
128 & 88  & \multicolumn{1}{c|}{8}     & 1429   & \multicolumn{1}{c|}{8}    & 746   \\ \hline
128 & 96  & \multicolumn{1}{c|}{6}     & 16     & \multicolumn{1}{c|}{8}    & 6227  \\ \hline
128 & 104 & \multicolumn{1}{c|}{4}     & 84     & \multicolumn{1}{c|}{4}    & 328   \\ \hline
128 & 112 & \multicolumn{1}{c|}{4}     & 1132   & \multicolumn{1}{c|}{4}    & 1132  \\ \hline
128 & 120 & \multicolumn{1}{c|}{2}     & 132    & \multicolumn{1}{c|}{2}    & 132   \\ \hline
\end{tabular}
\vspace{0em}
\end{table}

Fig. \ref{FigN128CRC6} shows the BLER performance of CRC-polar codes with $N = 128$, where 5G-CRC6 polynomial is $g(x) = x^6 + x^5 + 1$ and 5G-CRC11 polynomial is $g(x) = x^{11} + x^{10} + x^{9} + x^{5} + 1$.
The normal approximation of the finite blocklength capacity is called the dispersion bound and obtained in \cite{FiniteBlocklengthCapacity}.
The MWUB is calculated by \eqref{union_bound_a} with the MWD in Table \ref{TabMWDN128CRC6}.
In Fig. \ref{FigN128CRC6}(a) with $R = 0.25$, we observe that the BLER performance of MWD sequence with $L=2$ is worse than that of polar sequence.
Then, as $L$ increases, the performance of MWD sequence is better than that of polar sequence.
Specifically, the performance of MWD sequence with $L=32$ approaches the MWUB in the high SNR region and shows about $0.19$dB and $0.30$dB performance gaps at BLER $10^{-5}$ compared with the dispersion bound and the performance of polar sequence, respectively.
The reason is that the $\left(128,32\right)$ CRC-polar code constructed by MWD sequence has better MWD shown in Table \ref{TabMWDN128CRC6} than that constructed by polar sequence.
In Fig. \ref{FigN128CRC6}(b) with $R = 0.5$, we observe that the performance of MWD sequence with $L = 32$ is better than the optimized eBCH-polar code in \cite{yuan2019polar} at $E_b/N_0 = 4$dB.
The CRC-polar codes constructed by MWD sequence shows $0.17$dB and $0.14$dB performance gains with $L=32$ and $L=128$ at BLER $10^{-5}$ compared with the polar sequence, respectively.
The performance of MWD sequence with $L = 128$ has about $0.11$dB and $0.22$dB performance gaps at BLER $10^{-4}$ compared with the random linear combination \cite{InformationSCL} and dispersion bound, respectively.
Though the performance of the MWD sequence is worse than that of the random linear combination, the MWD sequence is nested and can be used in practical easily.

Fig. \ref{FigN256} illustrates the BLER performance of polar codes and PAC codes constructed by different methods with $N = 256$, $K = 128$ and $L = 16$, where polar+ and PAC+ are the construction methods in \cite{NondreasingCodeMWD} for polar codes and PAC codes, respectively.
The precoding in all PAC codes is performed with polynomial coefficients $[1,0,1,1,0,1,1]$, which is identical to Ar{\i}kan's PAC codes \cite{PAC}.
PAC codes are decoded by the list decoding in \cite{PAClist}.
In Fig. \ref{FigN256}, we observe that both the polar code and the PAC code constructed by the MWD sequence have the optimum performance among the MWD sequence, the polar sequence and the method in \cite{NondreasingCodeMWD}.
Specifically, the PAC code constructed by MWD sequence has about $0.14$dB performance gain at BLER $10^{-3}$ compared with PAC+.
Hence, optimizing the MWD can improve the performance of concatenated polar codes.

\subsection{Simulation Results of ECBS Algorithm}

In this subsection, we provide the BLER performance of polar codes constructed by the proposed ECBS algorithm with different list sizes.

\begin{figure}[t]
\setlength{\abovecaptionskip}{0.cm}
\setlength{\belowcaptionskip}{-0.cm}
  \centering{\includegraphics[scale=0.68]{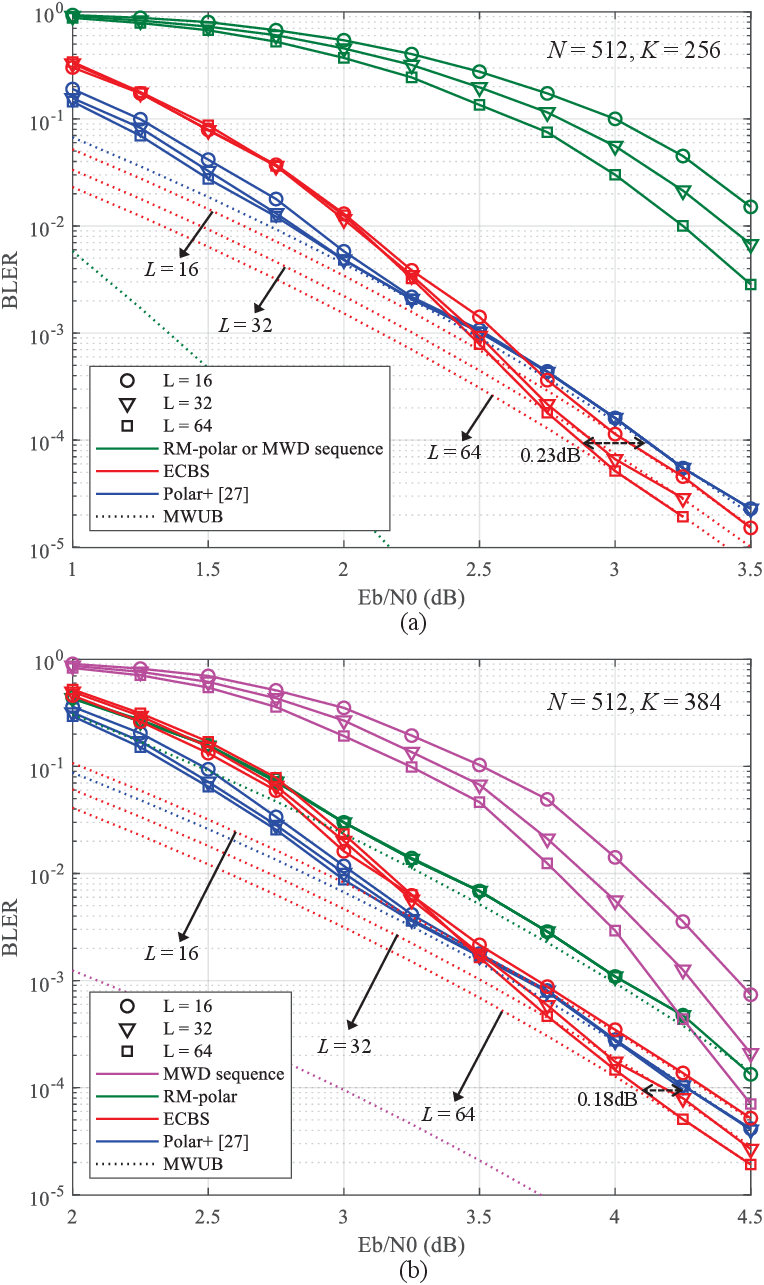}}
  \caption{Fig. \ref{FigN512}(a) and Fig. \ref{FigN512}(b) illustrate the BLER performance of $\left(512,256\right)$ and $\left(512,384\right)$ polar codes with different construction methods, respectively.}\label{FigN512}
  \vspace{-0em}
\end{figure}

\begin{table}[t]
\centering
  \caption{The MWDs of polar codes with different construction methods, where the designed $E_b/N_0$ of the ECBS algorithm are $1.5$dB and $2.57$dB for $\left(512,256\right)$ and $\left(512,384\right)$ polar codes, respectively.}\label{TabMWDN512}
\begin{tabular}{|c|cc|cc|}
\hline
                                                        & \multicolumn{2}{c|}{(512, 256)}      & \multicolumn{2}{c|}{(512, 384)}   \\ \hline
                                                        & \multicolumn{1}{c|}{$d_{\min}$} & $A_{d_{\min}}$    & \multicolumn{1}{c|}{$d_{\min}$} & $A_{d_{\min}}$ \\ \hline
\begin{tabular}[c]{@{}c@{}}ECBS ($L = 16$)\end{tabular} & \multicolumn{1}{c|}{16}   & 14432    & \multicolumn{1}{c|}{8}    & 16576 \\ \hline
\begin{tabular}[c]{@{}c@{}}ECBS ($L = 32$)\end{tabular}  & \multicolumn{1}{c|}{16}   & 9312     & \multicolumn{1}{c|}{8}    & 9408  \\ \hline
\begin{tabular}[c]{@{}c@{}}ECBS ($L = 64$)\end{tabular}  & \multicolumn{1}{c|}{16}   & 6240     & \multicolumn{1}{c|}{8}    & 6336  \\ \hline
Polar+\cite{NondreasingCodeMWD}                                                  & \multicolumn{1}{c|}{16}   & 18720    & \multicolumn{1}{c|}{8}    & 13504 \\ \hline
RM-polar                                      & \multicolumn{1}{c|}{32}   & 52955952 & \multicolumn{1}{c|}{8}    & 49344 \\ \hline
MWD sequence                                      & \multicolumn{1}{c|}{32}   & 52955952 & \multicolumn{1}{c|}{8}    & 192 \\ \hline
\end{tabular}
\vspace{0em}
\end{table}

Fig. \ref{FigN512} provides the BLER performance of $\left(512,256\right)$ and $\left(512,384\right)$ polar codes with different construction methods.
The MWDs of the polar codes in Fig. \ref{FigN512} are shown in Table \ref{TabMWDN512}.
The designed $E_b/N_0$ of the ECBS algorithm are $1.5$dB and $2.57$dB for $\left(512,256\right)$ and $\left(512,384\right)$ polar codes, respectively, and the information sets are provided in Appendix.
In Fig. \ref{FigN512}, the BLER performance of polar codes constructed by the ECBS algorithm approaches its MWUB in the high SNR region and the corresponding $A_{d_{\min}}$ is reduced as $L$ increases, which leads to the better performance with larger $L$.
In Fig. \ref{FigN512}(a), the information sets obtained by RM-polar \cite{RMpolarLi} and MWD sequence are identical.
Then, the RM-polar and the MWD sequence have the optimum MWUB among these construction methods but a large performance gap between the MWUB and the performance, since the entropy constraint is unsatisfied.
In Fig. \ref{FigN512}(b), the RM-polar satisfies the entropy constraint and its performance approaches the MWUB in the high SNR region.
However, the MWD sequence shows a large performance gap between the performance and the MWUB.
Hence, the entropy constraint is a metric to evaluate whether the performance of polar codes with limited list size under SCL decoding can approach the ML performance or not and the proposed ECBS algorithm can improve the MWD of polar codes to show better BLER performance as $L$ increases.
Finally, compared with the construction method in \cite{NondreasingCodeMWD}, the polar codes constructed by the ECBS algorithm have less $A_{d_{\min}}$ and show better performance in the high SNR region. Specifically, the $\left(512,256\right)$ and $\left(512,384\right)$ polar codes constructed by the ECBS algorithm with $L = 64$ have about $0.23$dB and $0.18$dB performance gains at BLER $10^{-4}$, respectively.

\section{Conclusion}

In this paper, the construction methods based on MWD are proposed to improve the performance of polar codes under SCL decoding. We first prove that the ML performance can approach the MWUB as the SNR goes to infinity. Then, we design the ordered and nested MWD sequence to apply fast construction without channel information and we prove that the MWD sequence is the optimum sequence evaluated by MWUB for polar codes obeying PO. Finally, we introduce the entropy constraint and propose the ECBS algorithm to construct polar codes. The simulation results show that the MWD sequence is suitable for constructing polar codes with short code length and the proposed ECBS algorithm can improve the performance of polar codes as the list size increases.

\vspace{-0em}

\appendix

The information sets for the $\left(512, 256\right)$ polar codes with $L = 16$, $32$ and $64$ in Fig. \ref{FigN512}(a) are as follows:
\begin{equation}
{\mathcal A}_{\tt L16} = {\mathcal A}_{\tt L32} \bigcup \left\{344,396,402\right\} \setminus \left\{109,167,271\right\},
\end{equation}
\begin{equation}
{\mathcal A}_{\tt L32} = {\mathcal A}_{\tt L64} \bigcup \left\{356,417\right\} \setminus \left\{94,155\right\},
\end{equation}
\begin{equation}
\begin{aligned}
&{\mathcal A}_{\tt L64} ={\mathcal A}_{\tt MWD} \bigcup \left\{
\begin{aligned}
&360,368,404,408,418,420,424,\\&432,449,450,452,456,464,480
\end{aligned}
\right\} \\
&\setminus \left\{
31,47,55,59,61,62,79,87,91,93,103,107,143,151
\right\},
\end{aligned}
\end{equation}
and the information sets for the $\left(512, 384\right)$ polar codes with $L = 16$, $32$ and $64$ in Fig. \ref{FigN512}(b) are as follows:
\begin{equation}
{\mathcal A}_{\tt L16} = {\mathcal A}_{\tt L32} \bigcup \left\{208,296,385\right\} \setminus \left\{46,53,75\right\},
\end{equation}
\begin{equation}
{\mathcal A}_{\tt L32} = {\mathcal A}_{\tt L64} \bigcup \left\{304,324\right\} \setminus \left\{51,71\right\},
\end{equation}
\begin{equation}
\begin{aligned}
{\mathcal A}_{\tt L64} =& {\mathcal A}_{\tt MWD} \bigcup \left\{
224,328,336,352,386,388,392,400
\right\} \\
&\setminus \left\{
15,23,27,29,30,39,43,45
\right\},
\end{aligned}
\end{equation}
where ${\mathcal A}_{\tt MWD}$ is the information set obtained by the MWD sequence.

\bibliographystyle{IEEEtran}
\bibliography{IEEEabrv,myrefs}

\begin{thebibliography}{10}
\providecommand{\url}[1]{#1}
\csname url@samestyle\endcsname
\providecommand{\newblock}{\relax}
\providecommand{\bibinfo}[2]{#2}
\providecommand{\BIBentrySTDinterwordspacing}{\spaceskip=0pt\relax}
\providecommand{\BIBentryALTinterwordstretchfactor}{4}
\providecommand{\BIBentryALTinterwordspacing}{\spaceskip=\fontdimen2\font plus
\BIBentryALTinterwordstretchfactor\fontdimen3\font minus
  \fontdimen4\font\relax}
\providecommand{\BIBforeignlanguage}[2]{{%
\expandafter\ifx\csname l@#1\endcsname\relax
\typeout{** WARNING: IEEEtran.bst: No hyphenation pattern has been}%
\typeout{** loaded for the language `#1'. Using the pattern for}%
\typeout{** the default language instead.}%
\else
\language=\csname l@#1\endcsname
\fi
#2}}
\providecommand{\BIBdecl}{\relax}
\BIBdecl

\bibitem{arikan}
E.~Ar{\i}kan, ``Channel polarization: A method for constructing
  capacity-achieving codes for symmetric binary-input memoryless channels,''
  \emph{{IEEE} Trans. Inf. Theory}, vol.~55, no.~7, pp. 3051--3073, Jul. 2009.

\bibitem{niuscl}
K.~Chen, K.~Niu, and J.~Lin, ``List successive cancellation decoding of polar
  codes,'' \emph{Electronics Lett.}, vol.~48, no.~9, pp. 500--501, 2012.

\bibitem{talvardyscl}
I.~Tal and A.~Vardy, ``List decoding of polar codes,'' \emph{{IEEE} Trans. Inf.
  Theory}, vol.~61, no.~5, pp. 2213--2226, May 2015.

\bibitem{3GPP_5G_polar}
3$^\text{rd}$ Generation Partnership Project~(3GPP), ``Multiplexing and channel
  coding,'' Tech. Rep. 3GPP TS 38.212 V15.0.0, Sophia Antipolis Cedex, France,
  2017.

\bibitem{DE}
R.~Mori and T.~Tanaka, ``Performance of polar codes with the construction using
  density evolution,'' \emph{{IEEE} Commun. Lett.}, vol.~13, no.~7, pp.
  519--521, Jul. 2009.

\bibitem{TalVardy}
I.~Tal and A.~Vardy, ``How to construct polar codes,'' \emph{{IEEE} Trans. Inf.
  Theory}, vol.~59, no.~10, pp. 6562--6582, Oct. 2013.

\bibitem{GA}
P.~Trifonov, ``Efficient design and decoding of polar codes,'' \emph{{IEEE}
  Trans. Commun.}, vol.~60, no.~11, pp. 3221--3227, Nov. 2012.

\bibitem{GA_DAI}
J.~Dai, K.~Niu, Z.~Si, C.~Dong, and J.~Lin, ``Does {Gaussian} approximation
  work well for the long-length polar code construction?'' \emph{IEEE Access},
  vol.~5, pp. 7950--7963, 2017.

\bibitem{PolarSpectrum}
K.~Niu, Y.~Li, and W.~Wu, ``Polar codes: Analysis and construction based on
  polar spectrum,'' \emph{arXiv:1908.05889}, Nov. 2019.

\bibitem{PolarSpectrumFastFading}
K.~Niu and Y.~Li, ``Polar codes for fast fading channel: Design based on polar
  spectrum,'' \emph{{IEEE} Trans. Veh. Technol.}, vol.~69, no.~9, pp.
  10\,103--10\,114, Sep. 2020.

\bibitem{PartialOrder}
C.~Sch{\"u}rch, ``A partial order for the synthesized channels of a polar
  code,'' in \emph{IEEE Int. Symp. Inform. Theory (ISIT)}, 2016, pp. 220--224.

\bibitem{PW}
G.~He \emph{et~al.}, ``Beta-expansion: A theoretical framework for fast and
  recursive construction of polar codes,'' in \emph{IEEE Global Communications
  Conference (GLOBECOM)}, 2017, pp. 1--6.

\bibitem{5GDesignPolar}
V.~Bioglio, C.~Condo, and I.~Land, ``Design of polar codes in {5G} new radio,''
  \emph{{IEEE} Commun. Surveys Tuts.}, vol.~23, no.~1, pp. 29--40, Firstquarter
  2021.

\bibitem{ConstructionAI}
L.~Huang, H.~Zhang, R.~Li, Y.~Ge, and J.~Wang, ``{AI} coding: Learning to
  construct error correction codes,'' \emph{{IEEE} Trans. Commun.}, vol.~68,
  no.~1, pp. 26--39, Jan. 2020.

\bibitem{ConstructionGenetic}
A.~Elkelesh, M.~Ebada, S.~Cammerer, and S.~t. Brink, ``Decoder-tailored polar
  code design using the genetic algorithm,'' \emph{{IEEE} Trans. Commun.},
  vol.~67, no.~7, pp. 4521--4534, Jul. 2019.

\bibitem{ConstructionRL}
Y.~Liao, S.~A. Hashemi, J.~M. Cioffi, and A.~Goldsmith, ``Construction of polar
  codes with reinforcement learning,'' \emph{{IEEE} Trans. Commun.}, vol.~70,
  no.~1, pp. 185--198, Jan. 2022.

\bibitem{RMpolar}
M.~Mondelli, S.~H. Hassani, and R.~L. Urbanke, ``From polar to {Reed-Muller}
  codes: A technique to improve the finite-length performance,'' \emph{{IEEE}
  Trans. Commun.}, vol.~62, no.~9, pp. 3084--3091, Sep. 2014.

\bibitem{DesShortPolarSCL}
V.~Miloslavskaya and B.~Vucetic, ``Design of short polar codes for {SCL}
  decoding,'' \emph{{IEEE} Trans. Commun.}, vol.~68, no.~11, pp. 6657--6668,
  Nov. 2020.

\bibitem{RecDesPolarSCL}
V.~Miloslavskaya, B.~Vucetic, Y.~Li, G.~Park, and O.-S. Park, ``Recursive
  design of precoded polar codes for {SCL} decoding,'' \emph{{IEEE} Trans.
  Commun.}, vol.~69, no.~12, pp. 7945--7959, Dec. 2021.

\bibitem{LinShuBook}
S.~Lin and D.~J. Costello, \emph{Error Control Coding (2nd ed.)}.\hskip 1em
  plus 0.5em minus 0.4em\relax PrenticeHall, Inc., 2004.

\bibitem{CRCdesign}
Q.~Zhang, A.~Liu, X.~Pan, and K.~Pan, ``{CRC} code design for list decoding of
  polar codes,'' \emph{{IEEE} Commun. Lett.}, vol.~21, no.~6, pp. 1229--1232,
  Jun. 2017.

\bibitem{ADSCL}
B.~Li, H.~Shen, and D.~Tse, ``An adaptive successive cancellation list decoder
  for polar codes with cyclic redundancy check,'' \emph{{IEEE} Commun. Lett.},
  vol.~16, no.~12, pp. 2044--2047, Dec. 2012.

\bibitem{dsliu}
Z.~Liu, K.~Chen, K.~Niu, and Z.~He, ``Distance spectrum analysis of polar
  codes,'' in \emph{IEEE Wireless Communications and Networking Conference
  (WCNC)}, 2014, pp. 490--495.

\bibitem{calculate_MWD}
M.~Bardet, V.~Dragoi, A.~Otmani, and J.-P. Tillich, ``Algebraic properties of
  polar codes from a new polynomial formalism,'' in \emph{IEEE Int. Symp.
  Inform. Theory (ISIT)}, 2016, pp. 230--234.

\bibitem{sphereMWD}
J.~Piao, K.~Niu, J.~Dai, and C.~Dong, ``Sphere constraint based enumeration
  methods to analyze the minimum weight distribution of polar codes,''
  \emph{{IEEE} Trans. Veh. Technol.}, vol.~69, no.~10, pp. 11\,557--11\,569,
  Oct. 2020.

\bibitem{PAC}
E.~Ar{\i}kan, ``From sequential decoding to channel polarization and back
  again,'' \emph{arXiv:1908.09594}, 2019.

\bibitem{NondreasingCodeMWD}
M.~Rowshan, S.~H. Dau, and E.~Viterbo, ``On the formation of min-weight
  codewords of polar/{PAC} codes and its applications,'' \emph{{IEEE} Trans.
  Inf. Theory}, pp. 1--1, 2023.

\bibitem{AWS1}
Y.~Li, H.~Zhang, R.~Li, J.~Wang, G.~Yan, and Z.~Ma, ``On the weight spectrum of
  pre-transformed polar codes,'' in \emph{IEEE Int. Symp. Inform. Theory
  (ISIT)}, 2021, pp. 1224--1229.

\bibitem{AWS2}
Y.~Li, Z.~Ye, H.~Zhang, J.~Wang, G.~Yan, and Z.~Ma, ``On the weight spectrum
  improvement of pre-transformed {Reed-Muller} codes and polar codes,'' in
  \emph{IEEE Int. Symp. Inform. Theory (ISIT)}, 2023, pp. 2153--2158.

\bibitem{PolarOpt0025}
J.~Piao, K.~Niu, J.~Dai, and C.~Dong, ``Approaching the normal approximation of
  the finite blocklength capacity within 0.025 db by short polar codes,''
  \emph{{IEEE} Wireless Commun. Lett.}, vol.~9, no.~7, pp. 1089--1092, Jul.
  2020.

\bibitem{yuan2019polar}
P.~Yuan, T.~Prinz, G.~Boecherer, O.~Iscan, R.~Boehnke, and W.~Xu, ``Polar code
  construction for list decoding,'' in \emph{International ITG Conference on
  Systems, Communications and Coding}, 2019, pp. 1--6.

\bibitem{KouniaBound}
A.~Cohen and N.~Merhav, ``Lower bounds on the error probability of block codes
  based on improvements on de {Caen}'s inequality,'' \emph{{IEEE} Trans. Inf.
  Theory}, vol.~50, no.~2, pp. 290--310, Feb. 2004.

\bibitem{InformationSCL}
M.~C. Co\c{s}kun and H.~D. Pf{\i}ster, ``An information-theoretic perspective
  on successive cancellation list decoding and polar code design,''
  \emph{{IEEE} Trans. Inf. Theory}, vol.~68, no.~9, pp. 5779--5791, Sep. 2022.

\bibitem{FiniteBlocklengthCapacity}
Y.~Polyanskiy, H.~V. Poor, and S.~Verdu, ``Channel coding rate in the finite
  blocklength regime,'' \emph{IEEE Transactions on Information Theory},
  vol.~56, no.~5, pp. 2307--2359, May 2010.

\bibitem{PAClist}
H.~Yao, A.~Fazeli, and A.~Vardy, ``List decoding of {A}r{\i}kan's {PAC}
  codes,'' in \emph{IEEE Int. Symp. Inform. Theory (ISIT)}, 2020, pp. 443--448.

\bibitem{RMpolarLi}
B.~Li, H.~Shen, and D.~Tse, ``A {RM}-polar codes,'' \emph{arXiv:1407.5483},
  2014.

\end{thebibliography}

\end{document}